\documentclass[journal]{IEEEtran}
 \usepackage{amsmath,amssymb}
 \usepackage{subfigure}
 \usepackage{graphicx,graphics,color,psfrag}
 \usepackage{cite,balance}
 \usepackage{caption}
 \allowdisplaybreaks
 \usepackage{algorithm}
 \usepackage{accents}
 \usepackage{amsthm}
 \usepackage{bm}
 \usepackage{algorithmic}
 \usepackage[english]{babel}
 \usepackage{multirow}
 \usepackage{enumerate}
 \usepackage{cases}
 \usepackage{stfloats}
 \usepackage{dsfont}
 \usepackage{color,soul}
 \usepackage{amsfonts}
 \usepackage{cite,graphicx,amsmath,amssymb}
 \usepackage{subfigure}
 \usepackage{fancyhdr}
 \usepackage{hhline}
 \usepackage{graphicx,graphics}
 \usepackage{array,color}
 \usepackage{amsmath}

\newtheorem{theorem}{Theorem}

\newtheorem{lemma}{Lemma}
\newtheorem{corollary}{Corollary}

\newtheorem{remark}{\bf Remark}

\def\E{\mathsf{E}}

\def\phi{\varphi}

\def\l{\left}
\def\r{\right}
\def\({\left(}
\def\){\right)}

\def\b0{{\mathbf{0}}}

\newcommand{\nn}{\nonumber}

\begin{document}

%

\title{\huge The Connectivity of Millimeter-Wave Networks in Urban Environments Modeled Using Random Lattices}

\author{Kaifeng Han, {\em{Student Member, IEEE}}, Ying Cui, {\em{Member, IEEE}} , Yueping Wu, {\em{Member, IEEE}}, and Kaibin Huang, {\em{Senior Member, IEEE}}

\thanks{
This work of K. Han and K. Huang was supported by Hong Kong Research Grants Council under the Grants 17209917, 17259416, F-HKU703/15T. The work of Y. Cui was supported by National Science Foundation of China under Grant 61401272 and Grant 61521062 as well as Shanghai Key Laboratory Funding STCSM15DZ2270400. Part of this work has been presented in IEEE Globecom 2017. The corresponding author is K. Huang.}

\thanks{ K. Han and K. Huang are with the Dept. of EEE at The  University of  Hong Kong, Hong Kong (e-mail: \{kfhan, haungkb\}@eee.hku.hk). Y. Cui is with the Department of EE, Shanghai Jiao Tong University, Shanghai 200240, China  (e-mail: cuiying@sjtu.edu.cn). Y. Wu is with Hong Kong Applied Science and Technology Research Institute (ASTRI), Hong Kong (e-mail: wu.yueping@gmail.com). }}
\maketitle

\begin{abstract}
\emph{Millimeter-wave} (mmWave) communication opens up tens of \emph{giga-hertz} (GHz) spectrum in the mmWave band for use by next-generation wireless systems, thereby solving the problem of spectrum scarcity. Maintaining connectivity stands out to be a key design challenge for mmWave networks deployed in urban regions due to the blockage effect characterising mmWave propagation.  In this paper, we set out to investigate the blockage effect on the connectivity of mmWave networks in a Manhattan-type urban region modeled using a random regular lattice while \emph{base stations} (BSs) are Poisson distributed in the plane. In particular, we analyse the connectivity probability that a typical user is within the transmission range of a BS and connected by a line-of-sight. First, we consider a single-tier network.  By jointly applying the random lattice and stochastic geometry theories, a lower bound on the connectivity probability is derived as a function of building parameters (e.g., size and site occupancy probability) and BS parameters (e.g., transmission range and BS density). For the case of dense buildings, the bound is derived in a simpler form. Next, the preceding lower bounds are tightened based on the geometric technique of partitioning the irregular blockage-free region around the typical user. Moreover, the analysis is generalized to mmWave channels with both LoS and NLoS paths. Last, the results are extended to a $K$-tier \emph{heterogeneous network} (HetNet), where building heights are random, and depending on its height, a building can block the signals transmitted by a subset of BS tiers but not all. The analysis shows that the connectivity probability of the $K$-tier HetNet increases linearly with the number of tiers. In general, our work quantifies the relation between the coverage of a mmWave network and the parameters of building and BS processes, providing useful guidelines for deploying practical networks in a Manhattan-type region.
\end{abstract}
\begin{IEEEkeywords}
Millimeter-wave networks, radio access networks, network connectivity, random lattice, stochastic geometry, wireless propagation.
\end{IEEEkeywords}

\section{Introduction}
Exploiting  the tens-of-GHz of available bandwidth in the \emph{millimeter-wave} (mmWave) band is embraced by both the industry and academia as a key solution for spectrum scarcity faced by 5G in view of the exponential traffic growth \cite{rappaport2013millimeter}. Consequently, mmWave communications are expected to play a key role in delivering extreme broadband access to ultra-dense mobile users in next-generation systems \cite{andrews2014will}. Though the physics of mmWave propagation is not yet fully understood, measurement results show that the main characteristic  of a mmWave channel is that signals are blocked (or at least severely attenuated) by objects  in an urban environment (e.g., buildings), known as the \emph{blockage effect}, which is much less severe in the microwave band below $6$-GHz \cite{anderson2004building, alejos2008measurement}.  In next-generation cellular networks, dense small-cell \emph{base stations} (BSs) will be deployed at flexible locations in-between or inside buildings, unlike macro-cell BSs installed typically on rooftops \cite{andrews2014will}. Then to operate the networks in the mmWave band, maintaining reliable connectivity for mobile users appears to be a key design challenge due to the blockage effect. Therefore, from the perspective of implementing  mmWave radio access networks, referred to simply as mmWave networks in this paper, it is important to quantify network connectivity based on a reasonable propagation model for the urban environment. To this end, we adopt the classic Manhattan-type  urban model where the spatial distributions of buildings follow  a random regular lattice and the stochastic geometry network model where BSs are randomly deployed following a \emph{Poisson point process} (PPP) for the case of a single-tier network or $K$-tier PPPs for the case of  a $K$-tier \emph{heterogeneous network} (HetNet). The level of connectivity in a mmWave network is measured using  the metric of   \emph{connectivity  probability}, defined   as the probability that a typical mobile has a \emph{line-of-sight} (LoS) link with at least one BS within a given transmission range. Then based on the said model, the connectivity probability is analyzed as a function of the  building density as well as other network parameters.

\subsection{Modeling Propagation in Millimeter-Wave  Networks}

Theoretic studies of traditional cellular networks are commonly based on the abstracted probabilistic channel models, e.g., Rayleigh or Rician fading, which are proposed based on rich scattering resulting from reflection and refraction properties for propagation with frequencies far below those of mmWaves  \cite{goldsmith2005wireless}. Such models are unsuitable for mmWave propagation where due to the said blockage effect, channels are dominated  by LoS links \cite{Andrews:2016aa}. Consequently,  to study the performance of mmWave networks, it is necessary to adopt propagation models that reflect the geometry and layout of blockage and scatter objects in the environment. Several well-known relevant modeling approaches are described as follows.

\begin{figure*}[t]
\centering
\includegraphics[width=13cm]{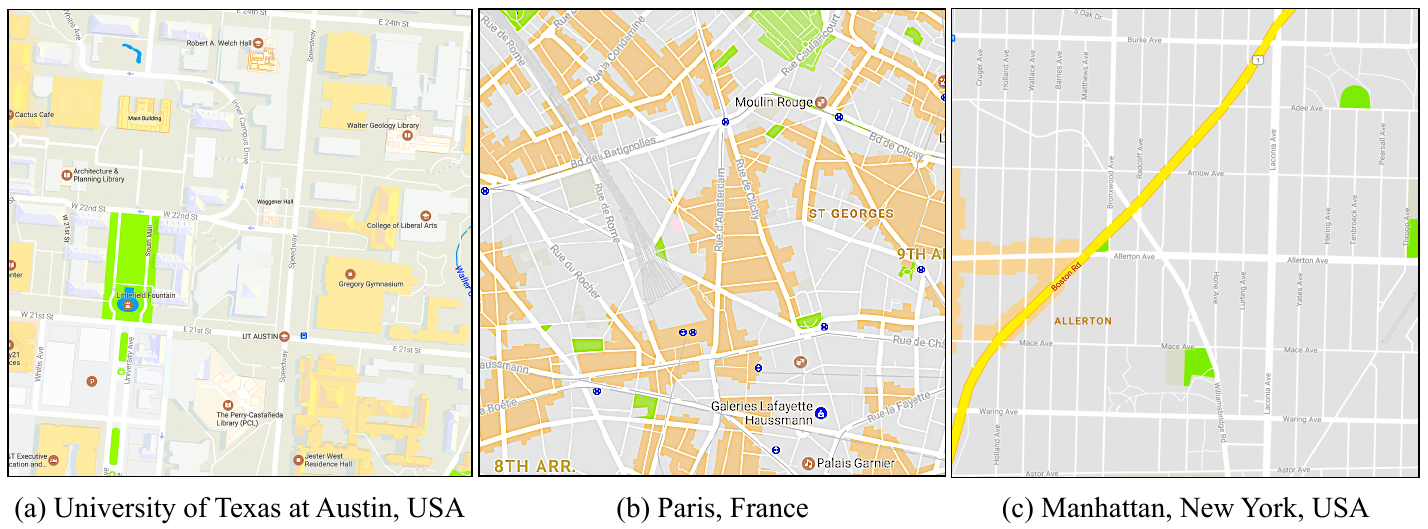}
\caption{Examples of urban regions: (a) University of Texas at Austin, USA,  (b) Paris, France and (c) Manhattan, New York, USA  that can be approximated  using the random-shape, Poisson-line  and random-lattice models, respectively.}\label{blockage_model_photo}
\label{Fig:Urban}
\end{figure*}

\subsubsection{Measurement Based  Models}
Perhaps the most accurate approach for studying  mmWave propagation is  \emph{ray tracing} that traces signal paths by simulation using a measurement based model accounting for  geometric properties  of objects in the propagation environment (e.g., locations, sizes, heights and orientations of buildings) or even their physical characteristics (e.g., surface materials) as well as atmospherical conditions \cite{schaubach1992ray,rizk1997two}. Though being a powerful tool for practical system design, ray tracing techniques do not yield mathematical tractability due to their high  complexity  and hence find little use in performance analysis of mmWave networks.

A simple measurement based model that can account for building blockage in mmWave propagation is one from the \emph{3GPP standard} where a random link belongs to either of the LoS or non-LoS types with given probabilities \cite{3gpptr}. The LoS and non-LoS probabilities can be fitted to a specific site  by measurement, e.g., the New York city \cite{akdeniz2014millimeter, sun2015path}. The distribution was found to be frequency independent for all bands up to $100$ GHz \cite{5gchannel},  making the model mmWave compatible.  Though being simple, the binary  channel-type  model is too coarse for depicting the detailed building layout needed for studying large-scale networks.

\subsubsection{Stochastic Geometry  Models}

Recently, modeling an urban region using stochastic geometry (e.g., Boolean model  or Poisson line process)   has emerged to be a promising approach for analyzing and designing large-scale mmWave networks for two main reasons \cite{bai2014analysis, baccelli2015correlated}. First, stochastic geometry provides a sufficiently elaborate and reasonable description of the random spatial distribution of blockage objects, accounting for their densities and geometric properties (e.g., shapes, orientations and sizes). Second, such an urban model can  be superimposed with a classic stochastic geometry network model to yield a tractable  model for a mmWave network.   The application of stochastic geometry to model and design wireless networks is a well-established approach (see e.g., the survey in \cite{haenggi2009stochastic_blockage}). Specifically, the spatial distribution of network nodes can be  represented by a suitable choice from a wide range of spatial point processes such as Poisson spatial tessellations  for cellular networks \cite{28andrews2011tractable} or cluster point processes for ad hoc radios \cite{GantiHaenggi:OutageClusteredMANET:2009, KB1}. The superposition of stochastic geometry models for blockage objects and  network nodes (BSs and mobiles) in an urban region allows the application of stochastic geometry theory to quantify the effects of blockage on the performance of mmWave networks.

There exist two main  stochastic geometry urban models  in the performance analysis of mmWave networks. The first one is the \emph{Boolean  model} considered in  \cite{bai2014analysis} where blockage objects are rectangular distributed in a plane following the Boolean model and having  random sizes and orientations. The model is suitable for the type of urban regions similar to the campus of The University of Texas at Austin, USA as shown in Fig.~\ref{Fig:Urban}(a), characterized by random and irregular blockage distributions. Combining the model and a network model with Poisson distributed BSs, random shape theory and other stochastic-geometry tools are applied for characterising the mmWave-network coverage by analysing the signal-and-interference distributions. For tractable analysis, it is assumed in \cite{bai2014analysis}  that different channels are independent of each other. The assumption does not hold in practice since an object with a large  volume can block multiple  mmWave links simultaneously with nonzero probability, introducing \emph{correlated blockage} for nearby links. Most recently, the blockage spatial correlation is studied in \cite{Gupta2017macro} using the Boolean model where buildings are represented by random line segments, and furthermore how blockage affects the network¡¯s reliability is characterised, considering users' macro diversity gain.

The second urban model, namely \emph{Poisson-line model},  is applied  in \cite{baccelli2015correlated} for studying wireless  networks with shadowing (or blockage) in an urban environment, where the distribution of streets in the urban region can be represented by a Poisson-line process. Such a model is suitable for cities with randomly oriented streets such as Paris, France as illustrated in Fig.~\ref{Fig:Urban}(b). In the mmWave-network model building on the said urban model,  network nodes  located on and off the lines are considered as being  outdoor and indoor, respectively.  Two outdoor locations along  the same line are connected by a LoS or otherwise by a non-LoS (with blockage), thereby accounting for correlated blockage.  Based on the model, the spatial correlation of channel shadowing is analyzed, providing results for studying the  interference distribution and network coverage \cite{baccelli2015correlated}.

In addition,  there exists a simplified model, called \emph{LoS-ball model},  for blockage effect as developed in \cite{singh2015tractable} for use in stochastic geometry network models. In the blockage model, a link between a BS-mobile pair is assumed to have a LoS only if the separation distance is shorter then a given threshold. Compared with the random-shape and Poisson-line models, the current one substantially simplifies the network-performance analysis but at the cost of losing an elaborate geometric description of blockage objects.

\subsubsection{Random-Lattice Models}

The urban regions in some cities  are suitably modeled using random (square) lattices but not the previously discussed random-shape or Poisson-line processes. One example is Manhattan, New York, USA as illustrated in Fig.~\ref{Fig:Urban}(c), giving the random-lattice urban models the well-known  name of Manhattan-type models. In such models, each square cell of the lattice is occupied by an object (e.g., a building) with a given probability and the locations of different cells are \emph{independent and identically distributed} (i.i.d.).  In the area of wave propagation, the models have been widely  used for  representing scatterers in an urban region  and used to analyze the distributions of wave-propagation distances in a given direction for a fixed number of reflections \cite{franceschetti1999propagation_blockage, marano1999statistical_blockage, marano2005ray_blockage}. Moreover, the random-lattice model is also commonly applied in network modelling and performance analysis in communication. In \cite[Part II]{martin2012SGWCom}, a comprehensive study of the percolation as well the connectivity in random-lattice model is provided. Based on the Manhanttan grid mobility model, the performance of the routing protocols in a mobile ad hoc network is investigated in \cite{jay2008grid}. In \cite{Choi2017spatial}, the random-lattice process is leveraged to model the deployment of BSs to study the coverage probability in a spatially repulsive cellular network. However, the applications of random-lattices to  blockage modelling still remain an uncharted area due to the lack of a tractable approach for network-performance analysis. This motivates the current work on making pioneering contributions to the area.

\subsection{Analysing Network Connectivity}

The connectivity of spatially random ad hoc networks is a classic research area (see e.g., \cite{haenggi2009stochastic_blockage, mao2012towards, xue2004number}). Two radio nodes are connected if they are within each other's transmission range. Then an ad hoc network is connected if all nodes have a single connected cluster with a probability close to one \cite{gupta1999critical}. In a fully connected ad hoc network, any pair of nodes can communicate via multi-hop transmissions.  Network connectivity typically exhibits a  phase-transition phenomenon where the network is fully connected almost surely if a particular network parameter, e.g., node transmission power \cite{gupta2000capacity} or interference density \cite{dousse2005impact}, satisfies some requirements. Different network models and analytical approaches, such as stochastic-geometric model and theory \cite{mao2013connectivity} or random matrix theory  \cite{dasgupta2015new}, have been proposed for studying network connectivity. In addition, various aspects of network connectivity in ad hoc networks have been investigated in the existing literature \cite{Bettstetter2002Degree, Bettstetter2004adhocNet, Coon2015scalinglaw, Pete2016mobility}. $K$-connectivity, referring to the event that a node is connected with $K$ neighbours, is studied in \cite{Bettstetter2002Degree} for wireless multi-hop networks. Connectivity properties for ad hoc networks are studied in \cite{Bettstetter2004adhocNet} via the development of an analytical framework. Considering both local and global connectivities, the corresponding network-scaling laws are derived in \cite{Coon2015scalinglaw} for bounded networks so as to obtain guidelines useful for controlling network topologies. How mobility affects network connectivity is studied in \cite{Pete2016mobility} by applying the well-known random waypoint model for mobile nodes. Due to  the severe  blockage effect, maintaining connectivity has appeared to be a key challenge for designing next-generation mmWave networks, and is an area largely uncharted,  motivating the current work. A classic network model for studying connectivity essentially consists of a set of randomly distributed nodes. In contrary, we consider mmWave radio access networks and propose a more complex stochastic geometry model comprising buildings (blockage objects) distributed as a random lattice and multi-tier BSs distributed as independent PPPs. Moreover, the connectivity for the mmWave networks is defined such that a typical user is connected to the network if the user lies within the transmission range of at least one BS and they are connected by a LoS. The definition differs from the mentioned classic one for ad hoc networks.

\subsection{Our Contributions}
The models and network performance metric for this work are elaborated as follows. The mmWave network is assumed to be deployed in a Manhattan-type urban region. As mentioned,  the existing spatial blockage models, including the Boolean model \cite{bai2014analysis} and the Poisson-line model \cite{baccelli2015correlated}, are unsuitable. Thus the buildings are modeled using a random regular lattice, where the plane is partitioned into uniform square sites with the size representing the building size. The buildings are overlaid with Poisson distributed BSs. We consider both single-tier and multi-tier BSs, modeled as a single PPP and multiple independent PPPs, respectively. The analysis focuses on a typical outdoor user at the origin while the extension of the results to a randomly located user is discussed. Then the network performance is measured by the \emph{connectivity probability} which is defined as the probability that the typical user is within the transmission range of at least one BS and furthermore they are connected by a LoS (unblocked by buildings). Interference is omitted in the analysis for simplicity, which can be justified by the fact that blockage and directional beamforming enabled by mmWave jointly suppress interference and enable the operation of mmWave networks in the noise-limited (i.e., power-limited)\footnote{In next-generation ultra-dense mmWave networks, a user can be exposed to a few unblocked strong interferers. The current analysis on the distribution of unblocked BSs for the typical user can be extended to analyze unblocked interferers, thereby investigating  interference in network connectivity analysis.} regime \cite{andrews2014will, RAP20131, Andrews:2016aa, rappaport2013millimeter}. One can see that the connectivity probability depends on both the distributions of the buildings and BSs. By jointly applying the theories of random lattice and stochastic geometry, bounds on the connectivity probability are derived in closed form.

The main contributions of the current work are summarised as follows.
\begin{enumerate}
  \item First, we consider a single-tier network. Define the blockage-free region as the region around the typical outdoor user that is free of buildings. A closed-form lower bound on the connectivity probability can be derived by inner bounding the blockage-free region by a disk centered at the user with a random radius. Using the bound and Poisson distribution of BSs, the lower bound on the connectivity probability is derived as a function of building parameters (e.g., size and site-occupancy probability) and BS parameters (e.g., transmission range and BS density). For dense sites, an asymptotic lower bound on the connectivity probability is derived, which has a simple form than the non-asymptotic counterpart.
  \item Next, the preceding lower bounds for the single-tier network are tightened by finding a tighter inner bound of the said blockage-free region. The technique is to partition the region into multiple sub-regions and inner bound each by a sector. The union of the sectors gives the said tighter inner bound on the blockage-free region. Then tighter bounds of the connectivity probability are derived. Furthermore, the analysis is generalized to mmWave channels with both LoS and non-LoS (NLoS) paths.
  \item Last, the preceding results are extended to a $K$-tier HetNet, comprising $K$-tier BSs with varying transmission ranges and densities over the tiers. Buildings with varying heights are also considered. The building with a random height blocks the signals transmitted by a corresponding subset of BS tiers. It is shown that the connectivity probability of the $K$-tier HetNet increases linearly with the number of tiers when the connectivity probability for each tier is small.
\end{enumerate}


\section{Models and Metric}\label{system_model}

The models and network performance metric are described in the following sub-sections. The notation is summarised in Table~\ref{tab:para_system}.
\begin{table*}[t]

\caption{Summary of Notation}\label{tab:para_system}
\begin{center}
\vspace{-2mm}
\begin{tabular}{|c!{\vrule width 1.5pt}l|}
\hline
Notation&Meaning\rule{0pt}{3mm}\\
\hhline{|=|=|}
$s$& Area of site\\
\hline
$\lambda_s$& The density of sites\\
\hline
$(a\sqrt{s}, b\sqrt{s})$, $\mathcal{S}_{a,b}$& The coordinate of site, the $(a,b)$th site\\
\hline
$\Phi$, $\tilde{\Phi}$& The random lattice process of a uniform-height blockage model, a $K$-height blockage model\\
\hline
$p_b$, $p_b^{(k)}$& Site occupancy probability in a uniform-height blockage model, a $K$-height blockage model\\
\hline
$\Sigma$, $\Sigma^{(k)}$& The random region covered by buildings in a uniform-height blockage model, a $K$-height blockage model\\
\hline
$\Pi$, $\lambda_c$& The PPP modeling BSs in the single-tier network, its density\\
\hline
$\Pi^{(k)}$, $\lambda_c^{(k)}$& The PPP modeling BSs in the $k$th tier network, its density\\
\hline
$r_b$, $r_b^{(k)}$& Radius of BS's coverage region in the single-tier, the $k$th tier network\\
\hline
$U_0$& Central located typical user\\
\hline
$\mathcal{F}$, $\mathcal{F}^{(k)}$& Blockage-free region of typical user a uniform-height blockage model, a $K$-height blockage model\\
\hline
$\mathcal{B}$& Disk region\\
\hline
$N(\cdot)$& Number of sites (either fully or partially) covered by a disk $\mathcal{B}$\\
\hline
$p_c$, $p_c^{(k)}$, $\widehat{p}_c$& Connectivity probability for the single-tier, the $k$th tier network, the $K$-tier HetNet\\
\hline

\end{tabular}
\vspace{-6mm}
\end{center}
\end{table*}

\begin{figure}[t]
\centering
\subfigure[Single-tier Network]{\includegraphics[width=8.5cm]{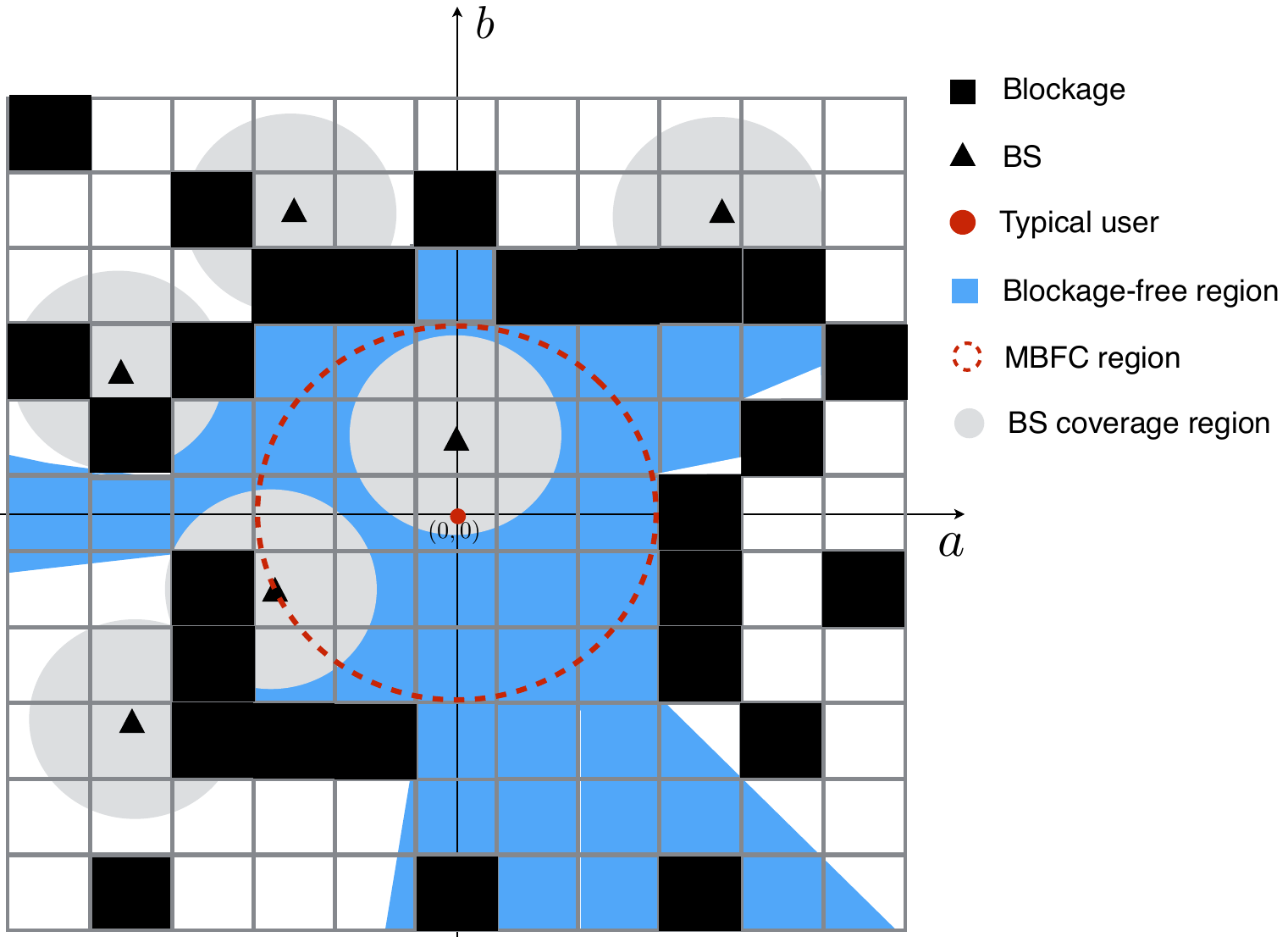}}
\hspace{20mm}
\subfigure[Two-tier Network]{\includegraphics[width=8.5cm]{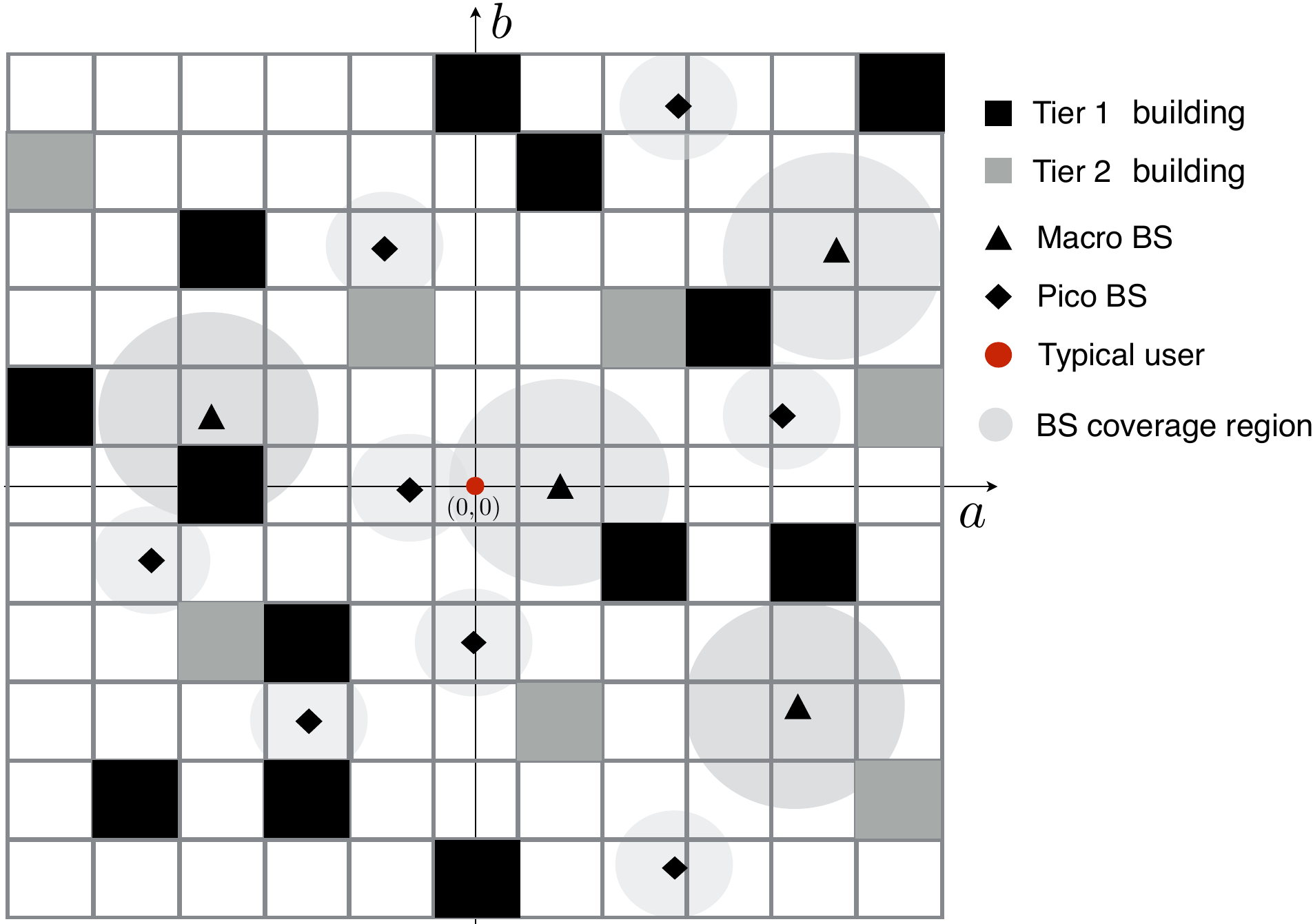}}
\caption{The spatial distribution of buildings, BSs, typical user in mmWave networks.}\label{sysmodel}
\end{figure}

\subsection{Blockage Spatial Distribution}

We model the mmWave blockage objects, namely buildings, in a Manhattan-type urban region using a random lattice process defined as follows. As illustrated in Fig.~\ref{sysmodel}, consider a regular lattice with density $\lambda_s$ that partitions the plane into uniform square areas of size $s = \frac{1}{\lambda_s}$, each called a site. For ease of expression, the center of each site is referred to as a lattice point. Let an arbitrary lattice point be the origin of plane $\mathds{R}^2$. Then the lattice points can be written as the set $\{(a\sqrt{s}, b\sqrt{s}) | (a,b) \in \mathds{Z}^2 \}$ and the $(a,b)$th site as $\mathcal{S}_{a,b}$, where $\mathds{Z}$ denotes the set of integers.

First, consider a uniform-height blockage model. That is, building heights are uniform. The random lattice, denoted as $\Phi$, represents a set of i.i.d. Bernoulli random variables $\Phi = \{ \mathcal{Z}_{a,b} \in\{0,1\}| (a,b) \in \mathds{Z}^2 \}$, where $\mathcal{Z}_{a,b} = 1$ and $0$ with probabilities $p_b$ and $\bar{p}_b = (1-p_b)$, respectively. Here, $\mathcal{Z}_{a,b}=1$ indicates that the site $\mathcal{S}_{a,b}$ is occupied by a building, and $\mathcal{Z}_{a,b}=0$ otherwise. The random region of plane $\mathds{R}^2$ covered by buildings is the  union of the  regions  of the sites in  $\{(a,b) \in \mathds{Z}^2 | \mathcal{Z}_{a,b}=1\}$, represented by $\Sigma$.

The above model can be extended to a $K$-height blockage model with densities varying for buildings with different random heights.  This model is more practical, as buildings in an urban region usually have different densities and heights. The $K$ different building heights are denoted by  $h^{(k)},\ k=1,\cdots, K$. Suppose  $h^{(1)}>\cdots>h^{(K)}>0$. For ease of notation, set $h^{(0)}=0$.    The effect of building heights on blockage is reflected in that buildings of heights $h^{(1)},\cdots, h^{(k)}$ can block signals transmitted by the $k$th tier BSs, as illustrated in the sequel. The random lattice of the $K$-height blockage model, denoted as $\tilde{\Phi}$, represents a set of i.i.d. random variables $\tilde{\Phi} = \{ \mathcal{Z}_{a,b} \in\{h^{(0)}, h^{(1)},\cdots, h^{(K)}\}| (a,b) \in \mathds{Z}^2 \}$, where  $\mathcal{Z}_{a,b} = h^{(k)} $ with probability  $p_b^{(k)}$, for all $k=1,\cdots, K$. Here, $\mathcal{Z}_{a,b}=h^{(0)}=0$ indicates that the site $\mathcal{S}_{a,b}$ is not occupied by a building, and $\mathcal{Z}_{a,b}=h^{(k)}>0$ indicates that $\mathcal{S}_{a,b}$ is occupied by a building of height $h^{(k)}$.  In addition, the probabilities $p_b^{(0)}$ and $p_b^{(k)}$ determine the density of empty sites and that of  buildings of height $h^{(k)}$. The random building region   corresponding to buildings of height $h^{(k)}$ is the union of the  regions  of the sites in  $\{(a,b) \in \mathds{Z}^2 | \mathcal{Z}_{a,b}=h^{(k)}\}$, represented by   $\Sigma^{(k)}$. Note that $\{\Sigma^{(k)}\}$ are correlated.


\subsection{Network Spatial Distribution}
First, consider a single-tier network comprising homogeneous BSs, as shown in Fig.~\ref{sysmodel}(a). The BS locations are modeled as a homogeneous PPP $\Pi = \{Y\}$ with density $\lambda_c$, where $Y \in \mathds{R}^2$ corresponds to the location of a particular BS. The BSs (or users) located in occupied sites can be treated as the \emph{indoor} BSs (or users) while others as the \emph{outdoor} BSs (or users). The network performance analysis focuses on a typical outdoor user, denoted by $U_0$, located at the origin and the extension to a randomly located user is subsequently discussed. The connectivity analysis for a typical indoor user is straightforward. Note that an indoor user cannot be connected with any BS outside its hosting building. Thus the corresponding connectivity probability is simply the probability that there exists at least one BS within the building, i.e., $1 - e^{-\lambda_c s}$. Therefore, we assume that there is no building at the origin, i.e., $\mathcal{Z}_{0,0} = 0$, and focus on the more complex analysis on the connectivity for a typical outdoor user.

Next, consider a HetNet comprising $K$ tiers of BSs as shown in Fig.~\ref{sysmodel}(b). The spatial distribution of BSs in the $k$th tier follows a homogeneous PPP with density $\lambda_c^{(k)}$, denoted as $\Pi^{(k)}$, where $k=1,\cdots, K$. Suppose $\lambda_c^{(1)}<\lambda_c^{(2)}<\cdots<\lambda_c^{(K)}$. The PPPs $\{\Pi^{(k)}\}$ are independent. All BSs in the same tier have the same transmission power. Suppose  the tier with a smaller index has a higher transmission power. That is, namely the 1st tier BSs have the largest transmission power while BSs in the $K$th tier have the smallest one. The BSs with a  higher transmit power (e.g., macro BSs) have a larger transmission range  and are overlaid by different classes of denser yet smaller coverage BSs (e.g., pico BSs or femto BSs). We consider the open-access strategy where any mobile user is allowed to connect to any BS tier without any restriction.\footnote{It is straightforward to extend the current analysis to the case of mixed open/closed-access. In this case, a BS within the connectivity range of the typical user can offer access to the user with a given probability or otherwise deny access, resulting in the thinning of the process of connectable BSs. The corresponding modification of the analysis is straightforward without affecting the main results and key insight.}

\subsection{Channel Model}\label{sec:channel}
Following \cite{3gpptr}, the channel between a BS and a user is of either LoS or NLoS, depending on whether it is intercepted by a building. For the case of NLoS, the complete blockage of signals is assumed in the whole paper except for Section \ref{sec:fadingSing} and \ref{sec:fadingHetnet} for simplification, reflecting severe propagation loss from penetrating or scattering by buildings \cite{Andrews:2016aa, bai2014analysis}.  On the other hand, for the case of LoS, the channel is assumed to have path-loss but no small-scale fading, as according to measurements of mmWave channels,  LoS paths are so dominant over reflection paths that small-scale fading  has a minor and negligible impact on mmWave communications \cite{andrews2014will, RAP20131, rappaport2013millimeter}. To be specific,  the BS transmitted power $P$ with distance $r$ is attenuated by the factor $r^{-\alpha_{\text{los}}}$ where $\alpha_{\text{los}} > 2$ denotes the path-loss exponent for LoS propagation. Considering the unit noise variance,  the corresponding receive \emph{signal-to-noise ratio} (SNR) is defined as $P_r := P r^{-\alpha_{\text{los}}}$.\footnote{In practical mmWave communication systems, sharp beamforming is typically deployed for enhancing link reliability. However, the  performance can be degraded due to transmit-and-receive beam misalignment.  Accounting for receive power loss due to random beam misalignment, the receive SNR can be modified as $\mathcal{Q} P_r$, where $\mathcal{Q} \in (0, 1)$ is a random variable \cite{Andrews:2016aa}.} The large-scale effects of shadowing are ignored in our model \cite{bai2014analysis, Geng2009mmwave} yet can be considered by applying the method in \cite{singh2015tractable} with the cost of tractability. Assume that the directional beamforming technique \cite{Macro2015mmwave} is implemented between the typical user and its associated BS using large-scale arrays. So both the typical user and its associated BS have the perfect channel knowledge and the BS can adaptively adjust its antenna's direction (i.e., steering oritentation) to obtain the maximal directionality gain.  As a desirable bonus, interference is assumed to be suppressed by the directional beamforming and also blocked by buildings.

We assume that a user is connected to the network if its receive SNR is above a given threshold $\theta$, i.e., $P_r \geq \theta$. It is worth mentioning that the current work mainly focuses on investigating LoS propagation and effects of blockage. The analysis is extended in Sections \ref{sec:fadingSing} and \ref{sec:fadingHetnet} to the case where channel comprises both LoS and NLoS paths in a single-tier and a $K$-tier HetNet, respectively. Considering only LoS propagation, two conditions have to be satisfied if a user connects to the network. First, the separation distance between the user and a BS should be shorter than a constant $r_b := (\frac{P}{\theta})^{\frac{1}{\alpha_{\text{los}}}}$, called the \emph{coverage (service) range} of the BS. Second, there has to be a LoS channel between them.

\subsection{Performance Metric}

The network performance is measured by the metric of connectivity probability $p_c$  defined as the probability that the typical user $U_0$ is connected to the network, i.e.,\footnote{Considering the randomness in beam misalignment, the connectivity probability can be expressed as $p_c = \Pr\{\mathcal{Q} P_r \geq \theta \} = \int_{0}^{1} \Pr\l\{P_r \geq \frac{\theta}{q} ~ |~ \mathcal{Q} = q\r\} f_\mathcal{Q}(q) {\rm d} q$,
where the conditional probability  $\Pr\l\{P_r \geq \frac{\theta}{q} ~|~ \mathcal{Q} = q\r\}$ can be obtained by  following  the identical procedure as that for investigating \eqref{defPc}, and  the probability density function of $\mathcal{Q}$, i.e., $f_\mathcal{Q}(q)$, can be characterised  (see e.g., \cite{Andrews:2016aa}).}
\begin{align}\label{defPc}
p_c = {\rm Pr} \left\{ P_r \geq \theta  \right\}.
\end{align}
Considering only LoS propagation in the single-tier network,  recall that the typical user $U_0$ is connected to the single-tier network if it is in the service range of at least one BS such that the BS is linked with $U_0$ by a LoS. Mathematically, the connectivity probability can be rewritten as follows. Let $\mathcal{B}(r_b)$ denote the disk centered at $U_0$ and with the radius $r_b$, where $r_b$ represents the coverage range of a BS in the single-tier network. Let $L(A,B) \triangleq \{ cA + (1-c)B | A \in \mathds{R}^2, B \in \mathds{R}^2, 0\leq c \leq 1\}$ denote the line segment connecting two points $A \in \mathds{R}^2$ and $B \in \mathds{R}^2$ in the plane. Define the  random \emph{blockage-free region} in the single-tier network as all points in the plane that are connected to $U_0$ by LoS, denoted by $\mathcal{F} \triangleq \{X \in \mathds{R}^2| L(X, U_0) \cap \Sigma = \emptyset\}$. Then, considering only LoS propagation, the connectivity probability $p_c$ in \eqref{defPc}  is equivalent to
\begin{align}\label{def:pc}
p_c = {\rm Pr} \left\{ \mathcal{F} \cap \mathcal{B}(r_b) \cap \Pi \neq \emptyset  \right\}.
\end{align}
Note that $\mathcal{F}$ and $\Pi$ depend on the blockage and BS distributions, respectively.

Consider the $K$-tier HetNet. Let $r^{(k)}_b$ represent the coverage radius of BSs in the $k$th tier.  Recalling the assumption that the tier with a smaller index has a higher transmission power, we have $r^{(1)}_b>r^{(2)}_b>\cdots >r^{(K)}_b$. Recall that  signals from a BS in the $k$th tier can be blocked by buildings of heights $h^{(1)}, \cdots, h^{(k)}$. Similar to the single-tier network scenario, the random blockage-free region for the $k$th tier is given by $\mathcal{F}^{(k)} = \{X \in \mathds{R}^2| L(X, U_0) \cap  (\cup_{\ell=1}^k\Sigma^{(\ell)}) = \emptyset\}$. Therefore, considering only LoS propagation,  the connectivity probability for the $k$th tier network is given by
\begin{align}\label{def:pc_k}
p_c^{(k)} = {\rm Pr} \left\{ \mathcal{F}^{(k)} \cap \mathcal{B}(r^{(k)}_b) \cap \Pi^{(k)} \neq \emptyset  \right\}.
\end{align}
Moreover, recalling the used open-access strategy, the event that the typical user covered by the $K$-tier HetNet is equivalent to that the set $\bigcup_{k=1}^{K} \left( \mathcal{F}^{(k)} \cap \mathcal{B}(r^{(k)}_b) \cap \Pi^{(k)}\right)$ is not empty. Then, for  the $K$-tier HetNet, the connectivity probability  $\widehat{p}_c$ can be written as
\begin{align}
\widehat{p}_c = {\rm Pr}\left\{\bigcup_{k=1}^{K} \left( \mathcal{F}^{(k)} \cap \mathcal{B}(r^{(k)}_b) \cap \Pi^{(k)}\right)\neq \emptyset\right\}.\label{eqn:K-HetNet-coverprob}
\end{align}
Note that $\mathcal{F}^{(k)}$ and $\Pi^{(k)}$ depend on the distributions of buildings of heights $h^{(1)},\cdots, h^{(k)}$ and the distribution of BSs in the $k$th tier, respectively.


\section{Connectivity of Single-tier Network}\label{single-tier}

It is challenging to derive the network connectivity probability due to the irregularity of the proposed blockage-free region. In this section, we derive its lower bounds with simple and insightful forms for the two cases of finite and high site densities by applying random-lattice and stochastic-geometry theories.

\subsection{Bounding Connectivity Probabilities}

\subsubsection{Finite Site Density}
Consider the case in which the site density $\lambda_s$ is finite. The case is equivalent to the one where each site has a finite area. For this case of finite $\lambda_s$, considering a typical outdoor mobile located at the origin for simplicity,  a lower bound on the connectivity probability $p_c$ defined in (\ref{def:pc}) is derived in this subsection.

First of all, some useful results are derived as follows.
\begin{lemma}[Counting Sites]\label{lem:num_site}\emph{
Given $r>0$, the number of sites (either fully or partially) covered by the disk $\mathcal{B}(r)$, denoted as $N(r)$,  is given as follows.
\begin{itemize}
\item[--]  For a finite ratio $\frac{r}{\sqrt{s}}$, $N(r)$ satisfies $N^{-}(r) \leq N(r) \leq N^{+}(r)$ with
\begin{align}
&N^{-}(r) = \l( 2\Big\lceil \frac{r}{\sqrt{2s}}-\frac{1}{2} \Big \rceil^{+} +1 \r)^2, \label{num_site_sr1}\\
&N^{+}(r) = \l( 2\Big\lceil \frac{r}{\sqrt{s}}-\frac{1}{2} \Big \rceil^{+} +1 \r)^2, \label{num_site_sr2}
\end{align}
where the operator $\lceil x \rceil^{+} = \max(\lceil x \rceil,0)$.
\item[--] For a large ratio $ \frac{r}{\sqrt{s}} \gg 1 $, $N(r)$ is given by
\begin{align}
N(r) = \frac{\pi r^2}{s} + O\l(\frac{r}{\sqrt{s}}\r). \label{num_site_lr}
\end{align}
\end{itemize}
}
\end{lemma}
Note that the lower and upper bounds on $N(r)$ given in \eqref{num_site_sr1} and \eqref{num_site_sr2} are derived by calculating the largest number of sites fully covered by $\mathcal{B}(r)$ and the smallest number of sites  fully covering $\mathcal{B}(r)$, respectively (see Fig.~\ref{sysmodel}(a)). In addition, the value of $N(r)$ in \eqref{num_site_lr} is obtained by ignoring boundary effects and focusing on the ratio between the area of $\mathcal{B}(r)$ and that of a site only.

Next, to derive a simple and insightful lower bound on $p_c$, we introduce the \emph{maximum (inscribed) blockage-free circular} (MBFC) region. Specifically, let
\begin{align}
R \triangleq  \max & \quad  r\nonumber\\
\mathbf{s.t.}& \quad r \in \left\{\sqrt{s}(n+\frac{1}{2}) \Big| n=0,1,\cdots\right\},\nonumber\\
& \quad \mathcal{B}(r) \cap \Sigma = \emptyset.
\end{align}
Note that $R \in \{\sqrt{s}(n+\frac{1}{2})| n=0,1,\cdots\}$ is a discrete random variable with the randomness included by the blockage region $\Sigma$. The discreteness of $R$ is due to that of $\Sigma$. The MBFC region is centered at the origin with radius $R$, denoted by $\mathcal{B}(R)$. The distribution of $R$ is given as follows.
\begin{lemma}[Distribution  of  MBFC Region]\label{lem:area_dist}
\emph{The \emph{probability mass function} (PMF) of $R$ is given by}\footnote{Note that $N(r_{n})=4(n+\frac{1}{2})^2-1$.}
\begin{align}\label{eq:S_dist}
{\rm Pr}\left\{R=r_n \right\}=\bar{p}_{b}^{4(n+\frac{1}{2})^2-1}-\bar{p}_{b}^{4(n+\frac{3}{2})^2-1}, n = 0, 1, \cdots,
\end{align}
\emph{where $r_n\triangleq \sqrt{s}(n+\frac{1}{2})$.}
\end{lemma}
\begin{proof}
See Appendix~\ref{proof:area_dist}.
\end{proof}

Using the above results, the connectivity probability $p_c$ can be lower bounded as follows. One can see that the MBFC region inner bounds the blockage-free region for the considered typical user: $\mathcal{B}(R) \subseteq \mathcal{F}$. Then replacing $\mathcal{F}$ with $\mathcal{B}(R)$ in the definition of $p_c$ in (\ref{def:pc}) gives:
\begin{align}
p_c &\geq {\rm Pr}\Big\{ \mathcal{B}(R) \cap \mathcal{B}(r_b) \cap \Pi \neq \emptyset  \Big\} \nn \\
&= {\rm Pr}\Big\{ \mathcal{B}\Big(\min(R,r_b)\Big) \cap \Pi \neq \emptyset \Big\} \label{pc_rb_lc}\\
&= \E\Big[ 1 - e^{-\pi \lambda_c (\min(R,r_b))^2} \Big] \nn \\
&= \Big(1-e^{-\pi\lambda_c r_{b}^{2}}\Big){\rm Pr}\Big(R\ge r_{b}\Big) \nn \\ &\hspace{0mm}+\sum_{n=0}^{\lfloor\frac{r_b}{\sqrt{s}}-\frac{1}{2}\rfloor^{+}}\Big(1-e^{-\pi\lambda_c r_n^{2}}\Big){\rm Pr}(R= r_n),\label{eq:CP1_ttl_thm_v2}
\end{align}
where the operator $\lfloor x \rfloor^{+} = \max(\lfloor x \rfloor,0)$.
Note that the two terms  in \eqref{eq:CP1_ttl_thm_v2} correspond to the cases of $R\geq r_b$ and $R <r_b$, respectively. In particular, in the case of $R <r_b$, $R$ takes the values in $\{r_n| r_n<r_b\}$.
On the other hand, leveraging the result in Lemma~\ref{lem:num_site}, we have
\begin{align}
{\rm Pr}\left(R\ge r_{b}\right) &= (1-p_{b})^{N(r_{b})-1} \geq (1-p_{b})^{N^{+}(r_{b})-1}. \label{ub_prb}
\end{align}
Substituting the distribution of $R$ in Lemma~\ref{lem:area_dist} and (\ref{ub_prb}) into (\ref{eq:CP1_ttl_thm_v2}) gives the following main result of this sub-section.
\begin{theorem}[Connectivity Probability for the Single-tier Network]\label{thm:CP1}
\emph{The connectivity probability for the single-tier network can be lower bounded as follows:}
\begin{align}\label{eq:CP1_dis}
&p_{c}\geq\Big(1-e^{-\pi\lambda_c r_{b}^{2}}\Big)\bar{p}_{b}^{N^{+}(r_b)-1} \nn \\
&\hspace{0mm}+ \sum_{n=0}^{\lfloor\frac{r_b}{\sqrt{s}}-\frac{1}{2}\rfloor^{+}}\Big(1 -e^{-\pi\lambda_c s (n+\frac{1}{2})^2}\Big)\bar{p}_{b}^{4(n+\frac{1}{2})^2-1}\Big( 1 - \bar{p}_{b}^{8(n+1)} \Big),
\end{align}
\emph{where   $N^{+}(r_b)$ is given in \eqref{num_site_sr1} and \eqref{num_site_sr2}.}
\end{theorem}
For the lower bound on $p_c$ in Theorem~\ref{thm:CP1}, the two terms correspond to the cases of $R \geq r_b$ and $R < r_b$, respectively. Note that  $r_b>r_n$ for $n=0,1,\cdots, \lfloor\frac{r_b}{\sqrt{s}}-\frac{1}{2}\rfloor^{+}$ implying $N^+(r_b)>N(r_n)=4(n+\frac{1}{2})^2-1$. Thus, the first term is dominant if the buildings are sparse, i.e., $p_b$ is small (or $\bar p_b$ is large), and the second term is dominant if the buildings are dense, i.e., $p_b$ is large (or $\bar p_b$ is small). Next, one thing can be observed from the result that the key parameters that determine the connectivity probability are the site-void probability $\bar{p}_b$, the BS density $\lambda_c$ and the coverage radius $r_b$. To be specific, $p_c$ is a monotone-increasing function of the two BS parameters ($\lambda_c$ and $r_b$) and also $\bar{p}_b$ when $\bar{p}_b$ is small. In addition, $p_c$ approaches one  exponentially fast  as $\lambda_c$ grows.

\subsubsection{High Site Density}
Consider the case of dense  sites ($\lambda_s \rightarrow \infty$). In this case, the discreteness of the blockage region $\Sigma$ varnishes such that the radius of the MBFC region, $R$, can be approximated as a continuous random variable with the following distribution:
\begin{align}
{\rm Pr}\{R\ge r\} = \bar{p}_b^{\lambda_s \pi r^2 + O(\sqrt{\lambda_s})}, \qquad \lambda_s \rightarrow \infty,
\end{align}
based on Lemma~\ref{lem:num_site}. Thus, as $\lambda_s  \rightarrow \infty$, the lower bound on $p_c$ in (\ref{eq:CP1_ttl_thm_v2}) can be written as:
\begin{align}\label{eq:CP1_ttl_smls_v2}
p_{c}&\geq\left(1-e^{-\pi\lambda_c r_{b}^{2}}\right)\bar{p}_b^{\pi \lambda_s  r_b^2}  \nn \\ &\hspace{0mm}+2\pi\int_0^{r_{b}}\left(1-e^{-\pi\lambda_c r^{2}}\right)\lambda_{s}\ln\frac{1}{\bar{p}_{b}}\bar{p}_{b}^{\lambda_{s}\pi r^{2}} r {\rm d}r,\nn\\
&= \frac{\lambda_c \left(1-\bar{p}_{b}^{\lambda_{s}\pi r_{b}^{2}}e^{-\pi\lambda_c r_{b}^{2}}\right)}{\lambda_c-\lambda_{s}\ln\bar{p}_{b}} + \frac{\lambda_{s}\ln\bar{p}_{b}\left(1-e^{ - \pi\frac{\lambda_c}{\lambda_s} }\right)}{\lambda_{s}\ln\bar{p}_{b}-\lambda_c}.
\end{align}
Applying the result given in (\ref{eq:CP1_ttl_smls_v2}) yields the following main result of the sub-section.
\begin{theorem}[Connectivity Probability for the Single-tier Network with a High Site Density]\label{thm:CP1_cont}
\emph{For the case of dense sites ($\lambda_s  \rightarrow \infty$), the connectivity probability for the single-tier network is lower bounded as follows:
\begin{align}
p_c\geq&\frac{1-e^{-\lambda_c \pi r_{b}^2(1 + \frac{\lambda_s}{\lambda_c}\ln \frac{1}{\bar{p}_b })}}{1+\frac{\lambda_s}{\lambda_c}\ln \frac{1}{\bar{p}_b }}+\frac{1- e^{ - \pi\frac{\lambda_c}{\lambda_s} }}{1+\frac{\lambda_c}{ \lambda_s\ln\frac{1}{\bar{p}_{b}}}}. \label{eq:CP1_cont_smls}
\end{align}}
\end{theorem}
Theorem~\ref{thm:CP1_cont} provides a closed-form lower bound on the connectivity probability for the single-tier network with dense sites. One observation can be made from the result in (\ref{eq:CP1_cont_smls}) is that for given BS density $\lambda_c$, the lower bound on $p_c$ decreases to zero as the building density $\lambda_s$ grows to infinity due to the blockage effect. On the other hand,  the effect can be counteracted by increasing the BS density $\lambda_c$, as the lower bound on $p_c$ depends on the ratio between $\lambda_c$ and $\lambda_s$ in the high building density case. The rule of thumb on the required density can be obtained from (\ref{eq:CP1_cont_smls}) as follows. Assuming that the site occupancy probability $p_b$ is small, we have
\begin{align}
\ln \frac{1}{\bar{p}_b } = \ln \frac{1}{1 - p_b } \approx p_b. \label{logpb}
\end{align}
Substituting (\ref{logpb}) into (\ref{eq:CP1_cont_smls}) yields
\begin{align}
p_c\geq&\frac{1-e^{-\lambda_c \pi r_{b}^2(1 + \frac{\lambda_s p_b}{\lambda_c})}}{1+\frac{\lambda_s p_b}{\lambda_c}}+\frac{1- e^{ - \pi\frac{\lambda_c}{\lambda_s} }}{1+\frac{\lambda_c}{ \lambda_sp_b}}.
\end{align}
Then a required value of $p_c$ can be guaranteed if the ratio $\frac{\lambda_s p_b}{\lambda_c}$ is fixed, for small $p_b$.

\subsection{Tightening the Bounds on Connectivity Probabilities}

\begin{figure}[t]
\centering
{\includegraphics[width=8.5cm]{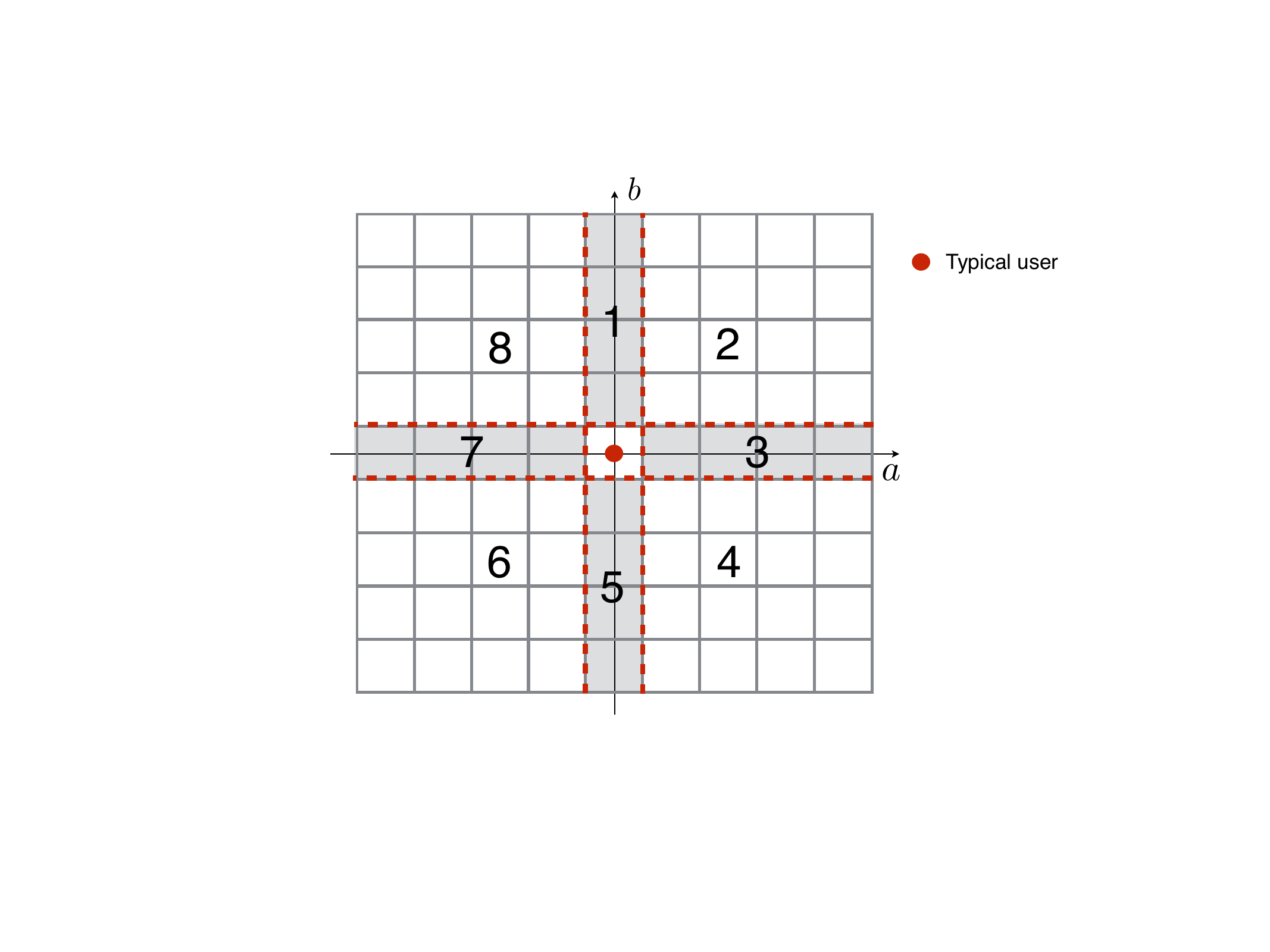}}
\caption{Dividing the urban region into eight non-overlapping regions, which are indexed by $n \in \{1,2,\cdots,8\}$.}\label{multi_region}
\end{figure}

Building on the analytical method developed in the preceding sub-section, a more complex one is developed in this sub-section to obtain lower bounds on the connectivity probability for finite and high site densities that are tighter than those in the preceding sub-section. The key idea is to define an inner bound of the blockage-free region $\mathcal{F}$ that is tighter than the MBFC region $\mathcal{B}(R)$ defined in the preceding sub-section. One promising way is to partition the plane into multiple non-overlapping regions and define the corresponding non-overlapping blockage-free regions, whose union yields the desired inner bound of $\mathcal{F}$. Given the typical outdoor user at the origin, we consider one particular partition of $\mathds{R}^2 \backslash \mathcal{S}_{0,0}$ that comprises eight non-overlapping regions (with an entire site belonging to one region), denoted as $\mathds{R}^2_{(1)}, \mathds{R}^2_{(2)}, \cdots, \mathds{R}^2_{(8)}$, as illustrated in Fig.~\ref{multi_region}. Note that this partition ensures that $\Sigma \cap \mathds{R}^2_{(n)}$, $n = 1,2,\cdots, 8$ are independent. In addition, the eight regions can be classified into two groups: regions indexed by the  $\{1,3,5,7\}$ and regions indexed by the $\{2,4,6,8\}$. The geometric characteristics of all regions in the same group are the same, and the characteristics of the two groups are different. Let
\begin{align}
R_n \triangleq \max & \quad  r\nonumber\\
\mathbf{s.t.}& \quad r \in \left\{\sqrt{s}(m+\frac{1}{2}): m =0,1,\cdots\right\},\nonumber\\
&\quad \mathcal{B}(r) \cap \Sigma \cap \mathds{R}^2_{(n)} = \emptyset.
\end{align}
Then, $\mathcal{B}(R_n) \cap \mathds{R}^2_{(n)}$ is the blockage-free region in $\mathds{R}^2_{(n)}$ with the largest radius. Notice that due to the independence of $\Sigma \cap \mathds{R}^2_{(n)}$, $n = 1,2,\cdots, 8$, $R_n$, $n = 1,2,\cdots, 8$ are independent and so are $\mathcal{B}(R_n) \cap \mathds{R}^2_{(n)}$, $n = 1,2,\cdots, 8$.
One can see that, the combined region $\bigcup_n \Big( \mathcal{B}(R_n) \cap \mathds{R}^2_{(n)} \Big)$ has more geometric degrees-of-freedoms in approaching $\mathcal{F}$ than $\mathcal{B}(R)$, yielding a tighter inner bound for $\mathcal{F}$. That is, we have
\begin{align}
\mathcal{B}(R) \subseteq  \cup_n \Big( \mathcal{B}(R_n) \cap \mathds{R}^2_{(n)} \Big)  \subseteq \mathcal{F}.
\end{align}
Then replacing $\mathcal{F}$ with $\bigcup_n \Big( \mathcal{B}(R_n) \cap \mathds{R}^2_{(n)} \Big)$ in the definition of the connectivity probability given in (\ref{def:pc}) leads to a tighter lower bound than that in \eqref{pc_rb_lc}:
\begin{align}
p_c &\geq {\rm Pr}\Big\{ \cup_n \Big( \mathcal{B}(R_n) \cap \mathds{R}^2_{(n)} \Big) \cap \mathcal{B}(r_b) \cap \Pi \neq \emptyset \Big\} \nn \\
&=1 - {\rm Pr}\l\{ \cup_n \Big( \mathcal{B}(R_n) \cap \mathds{R}^2_{(n)} \Big) \cap \mathcal{B}(r_b) \cap \Pi = \emptyset \r\} \nn \\
&\overset{(a)}{=} 1 - \prod_{n=1}^{8} {\rm Pr}\Big\{ \mathcal{B}(R_n) \cap \mathcal{B}(r_b) \cap \mathds{R}^2_{(n)} \cap \Pi = \emptyset \Big\}\nn \\
&=1 - \prod_{n=1}^{8}\Big( 1 - \tilde{p}_c^{(n)} \Big), \label{pc_8regions}
\end{align}
where
\begin{align}\label{pcn}
\tilde{p}_c^{(n)} = {\rm Pr}\Big\{  \mathcal{B}( \min(R_n, r_b) ) \cap \mathds{R}^2_{(n)} \cap \Pi \neq \emptyset \Big\}
\end{align}
represents the connectivity probability in region $\mathds{R}^2_{(n)}$ and (a) is due to the independence of $\mathcal{B}(R_n) \cap \mathds{R}^2_{(n)}\cap \Pi$, $n = 1,2,\cdots, 8$.

Using (\ref{pc_8regions}), a closed-form expression for the tighter lower bound on $p_c$ can be derived following a similar procedure as in the preceding sub-section. To this end, one can observe the similarity in the expressions for $\tilde{p}_c^{(n)}$ in (\ref{pcn}) and the lower bound on $p_c$ in \eqref{pc_rb_lc}. Then adopting a similar procedure as for deriving Theorem~\ref{thm:CP1} results in Lemma~\ref{lem:pc_8regions} in the sequel. The details are omitted for brevity.

\begin{lemma}\label{lem:pc_8regions}
\emph{The connectivity probabilities $\{ \tilde{p}_c^{(n)} \}$ defined in (\ref{pcn}) can be bounded as follows.
\item[--]  For $n =$ 1, 3, 5, and 7, $\tilde{p}_c^{(n)} \geq q^{(n)}$ where}
\begin{align}
q^{(n)} &= \Bigg( 1 - \exp\Big( - s \lambda_c \Big\lfloor \frac{r_b}{\sqrt{s}} -\frac{1}{2} \Big\rfloor^{+} \Big)\Bigg)\bar{p}_{b}^{\lceil \frac{r_b}{\sqrt{s}} -\frac{1}{2} \rceil^{+}} \nn \\
&\hspace{0mm}+ \sum_{\ell=0}^{\lfloor\frac{r_b}{\sqrt{s}}-\frac{1}{2}\rfloor^{+}}p_b \Big( 1 - e^{-s \lambda_c \ell} \Big)\bar{p}_{b}^{\ell}. \label{eqn:multi-q-n-1}
\end{align}
\emph{\item[--] For $n =$ 2, 4, 6, and 8, $\tilde{p}_c^{(n)} \geq q^{(n)}$ where
\begin{align}
q^{(n)} &=  \Bigg( 1 - \exp\Big( - \frac{1}{4}\pi s \lambda_c \Big(\Big\lfloor \frac{r_b}{\sqrt{s}} -\frac{1}{2} \Big\rfloor^{+} \Big)^2 \Big)\Bigg)\bar{p}_{b}^{\l( \lceil \frac{r_b}{\sqrt{s}}-\frac{1}{2} \rceil^{+}\r)^2} \nn \\
&\hspace{0mm}+ \sum_{\ell=0}^{\lfloor\frac{r_b}{\sqrt{s}}-\frac{1}{2}\rfloor^{+}}\Big( 1 - e^{- \frac{1}{4}\pi s \lambda_c \ell^2} \Big)\bar{p}_{b}^{\ell^2} \Big(1- \bar{p}_{b}^{2\ell+1} \Big).\label{eqn:multi-q-n-2}
\end{align}}
\end{lemma}
Note that in \eqref{eqn:multi-q-n-1} for $n=1,3,5,7$, $ s \Big\lfloor \frac{r_b}{\sqrt{s}} -\frac{1}{2} \Big\rfloor^{+}$  is the area of $\mathcal{B}(r_b) \cap \mathds{R}^2_{(n)}$ and $s \ell$ is a lower bound on the area of  $\mathcal{B}(r_{\ell}) \cap \mathds{R}^2_{(n)}$; in \eqref{eqn:multi-q-n-2} for $n=2,4,6,8$, $\frac{1}{4}\pi s  \Big(\Big\lfloor \frac{r_b}{\sqrt{s}} -\frac{1}{2} \Big\rfloor^{+}\Big)^2$   is the area of $\mathcal{B}(r_b) \cap \mathds{R}^2_{(n)} $ and $\frac{1}{4}\pi s  \ell^2$ is a lower bound on the area of  $\mathcal{B}(r_\ell) \cap \mathds{R}^2_{(n)}$.

Substituting Lemma~\ref{lem:pc_8regions} into (\ref{pc_8regions}) yields the following main result of this sub-section, which improves the result in Theorem~\ref{thm:CP1}.
\begin{theorem}\label{thm:CP1_8regions}
\emph{An alternative lower bound on the network-connectivity probability for the single-tier network is
\begin{align}
p_c \geq 1 - \prod_{n=1}^{8} \Big(1 - q^{(n)}\Big), \label{ieq:cover_8region}
\end{align}
where $\{ q^{(n)} \}$ are given in Lemma~\ref{lem:pc_8regions}.}
\end{theorem}
Consider the case where $\{ \tilde{p}_c^{(n)} \}$ are small. Then it follows from \eqref{pc_8regions} and Lemma~\ref{lem:pc_8regions} that
\begin{align}
p_c \geq \sum^{8}_{n=1} \tilde{p}_c^{(n)}\geq \sum_{n=1}^{8} q^{(n)}.
\end{align}
The summation reflects the improvement on the tightness of the lower bound on $p_c$ with respect to that in Theorem~\ref{thm:CP1}. That is, in the case of small $\{ \tilde{p}_c^{(n)} \}$, it can be seen more clearly that the result in Theorem~\ref{thm:CP1_8regions} improves that in Theorem~\ref{thm:CP1}.

Consider the network with high site density ($\lambda_s \rightarrow \infty$). The areas of the regions  $\mathds{R}^2_{(1)}, \mathds{R}^2_{(3)}, \mathds{R}^2_{(5)}$ and $\mathds{R}^2_{(7)}$ are close to zero when the site density $\lambda_s$ is sufficiently large (i.e., the site area $s$ is sufficiently small), which can be easily seen from Fig.~\ref{multi_region}, and thus can be ignored for ease of analysis. Following a similar procedure presented in the preceding sub-section, an asymptotic lower bound on the connectivity probability is derived as follows.
\begin{theorem}\label{thm:CP1_cont_8region}
\emph{For the case of high site density ($\lambda_s \rightarrow \infty$), an alternative lower bound on the connectivity probability for the single-tier network is}
\begin{align}
p_c \geq 1 - (1 - \tilde{p}_c)^4,
\end{align}
\emph{where}
\begin{align}\label{eq:CP1_cont_smls_8region}
\tilde{p}_{c} &= \frac{1-\bar{p}_{b}^{\frac{1}{4}\lambda_{s}\pi r_{b}^{2}}e^{-\frac{1}{4}\pi\lambda_c r_{b}^{2}}}{1-\frac{\lambda_{s}}{4\lambda_c}\ln\bar{p}_{b}}+\frac{1-e^{-\frac{\pi}{4}\frac{\lambda_c}{\lambda_s}}}{1-\frac{4\lambda_c}{\lambda_{s}\ln\bar{p}_{b}}}.
\end{align}
\end{theorem}

Similarly, the result in Theorem~\ref{thm:CP1_cont_8region} improves that in Theorem~\ref{thm:CP1_cont}. For the case where $\tilde{p}_c$ is small, it follows from Theorem~\ref{thm:CP1_cont_8region} that $p_c \geq 4 \tilde{p}_c$, where the factor 4 arises from the number of the main regions $\mathds{R}^2_{(2)}, \mathds{R}^2_{(4)}, \mathds{R}^2_{(6)}, \mathds{R}^2_{(8)}$ in the partition of the plane (see Fig.~\ref{multi_region}). The factor 4 reflects the tightening of the lower bound on $p_c$. In general, increasing the number of regions in the partition leads to a more accurate approximation of the blockage-free region (see Fig.~\ref{sysmodel}(a)) and hence an increasingly tighter lower bound on $p_c$.

\subsection{Randomly Located Typical User }
The preceding analysis assumes a typical outdoor user at the origin for ease of notation and expression. Extending the results to the general case of a randomly located user is straightforward. To incorporate the effect due to the random offset of the typical user from the origin in a tractable manner, $r$ in the upper bound of   Lemma~\ref{lem:num_site}  is replaced with $r+\frac{1}{2}\sqrt{s}$, while  $r$ in the lower bound of   Lemma~\ref{lem:num_site}  is replaced with $r-\frac{1}{2}\sqrt{2s}$.
In other words, for the case of a randomly located user, $N^{-}(r)$ and $N^{+}(r)$ in \eqref{num_site_sr1} and \eqref{num_site_sr2} can be replaced with
\begin{align}\label{num_sit_random}
&N^{-}(r) = \l( 2\Big\lceil \frac{r}{\sqrt{2s}}-1 \Big \rceil^{+} +1 \r)^2, \nn \\
&N^{+}(r) = \Bigg( 2\Big\lceil \frac{r}{\sqrt{s}} \Big \rceil^{+} +1 \Bigg)^2.
\end{align}
It is straightforward to modify other analytical results accordingly, without changing the key insight. For example, the result in Theorem \ref{thm:CP1} can be modified using  $N^{+}(r)$ in \eqref{num_sit_random}.

\subsection{Extension to Channel Model with Both LoS and NLoS Paths}\label{sec:fadingSing}
Recall that in the previous analysis, we assume that the signals are completely blocked by buildings and hence NLoS paths are ignored and the focus is on LoS paths. Some experiments show that NLoS paths, although non-dominant, also exist in mmWave channels and can provide  connectivity  in the absence of LoS paths (i.e., LoS BSs) \cite{akdeniz2015mmwaveChannel, Macro2015mmwave, tatino2017beam}. Therefore, in this subsection, we extend the channel model in Section \ref{sec:channel} to incorporate both LoS and NLoS paths. It is worth mentioning that in most existing works such as \cite{Andrews:2016aa, bai2014analysis}, the receive signals are assumed to be spatially separated into LoS and NLoS signals and characterised using the same channel model but with different parameters (e.g., path loss exponents, values of Nakagami small-scale fading, etc.). Similarly, we assume LoS signals and NLoS signals are spatially  separated. However, different from existing works \cite{Andrews:2016aa, bai2014analysis}, we adopt the proposed LoS channel model in Section \ref{sec:channel} to characterise LoS signals, which is more practical and accurate, as illustrated in Sections \ref{sec:fadingSing} and \ref{sec:fadingHetnet}. In addition, for analytical tractability, as in \cite{Andrews:2016aa, bai2014analysis}, we assume rich scattering for  NLoS  channels and adopt the classic NLoS propagation model for characterising NLoS signals under this assumption. To be specific, the channel attenuation for the NLoS channel model is modelled as $G r^{-\alpha_{\text{nlos}}}$ (i.e., the receive SNR is $P_r = P G r^{-\alpha_{\text{nlos}}}$), where $G$ is a random variable modelling small-scale fading, $r$ is the propagation distance,  and $\alpha_{\text{nlos}} > 2 $ (usually larger than $\alpha_{\text{los}}$) is the path-loss exponent. It is important to note that the model is much simpler than the  LoS propagation model since the channel gain no longer depends on the specific locations of scatterers (or buildings).\footnote{ Thus, it is unnecessary to introduce a ``blockage-free region" for deriving the  connectivity probability when considering NLoS paths.}

Now, we analyze the connectivity probability under the generalized channel model with both LoS and NLoS paths. The typical user can connect to the network via either LoS or NLoS paths. Recalling that LoS and NLoS paths are assumed to be spatially separated, the connectivity probability can be written as
\begin{align}\label{pc_twoterms}
p_c &= 1 - \l(1 - p_c(\text{LoS}) \r)  \l(1 - p_c(\text{NLoS}) \r),
\end{align}
where $p_c(\text{LoS})$ ($p_c(\text{NLoS})$) denotes the connectivity probability with only LoS (NLoS) paths. The expression of $p_c(\text{LoS})$ is given in \eqref{def:pc} and the expression of  $p_c(\text{NLoS})$ is given by
\begin{align}
p_c(\text{NLoS})
&= {\rm Pr}\l\{\max_{Y \in \Pi} PG_Y |Y|^{-\alpha_{\text{nlos}}} \geq \theta \r\}.\label{Eq:Pc:NLoS}
\end{align}

The result of $p_c(\text{LoS})$ has been given in Theorem~\ref{thm:CP1}. Now, we calculate $p_c(\text{NLoS})$ defined in \eqref{Eq:Pc:NLoS} by applying the theory of Marked PPP  (see e.g., \cite[Chapter 7]{martin2012SGWCom}).   First, given the BS process $\Pi$ and using fading coefficients $\{G_Y\}$ as their marks, a marked PPP for BSs, denoted as $\tilde{\Pi}$,  can be defined as
\begin{align}\label{Pi_bs}
\tilde{\Pi} = \l\{ (Y, G_Y) \in \Pi\times \mathds{R}^+ ~|~ PG_Y  |Y|^{-\alpha_{\text{nlos}}} \geq \theta \r\}.
\end{align}
By using the Marking Theorem \cite[Chapter 7]{martin2012SGWCom}, the intensity measure of $\tilde{\Pi}$, denoted as $\mu(\tilde{\Pi})$, is derived as
\begin{align}\label{density_bs}
\mu(\tilde{\Pi}) = 2 \pi \lambda_c \int_{0}^{\infty} {\rm Pr} \l\{ G_Y \geq \frac{\theta}{P} r^{\alpha_{\text{nlos}}}  \r\}r {\rm d}r,
\end{align}
where $G_Y$ is a random variable following a specific fading distribution, e.g., Rayleigh fading \cite{bai2014analysis}, Nakagami fading \cite{Andrews:2016aa}, or Log-Normal fading in \cite{Macro2015mmwave}. In particular,  for the simple case of Rayleigh fading, i.e., $G_Y \sim \exp(1)$, we have
\begin{align}\label{density_rayfading}
\mu(\tilde{\Pi})  = \pi \lambda_c \int\limits^{\infty}_{0} e^{- \frac{\theta t^{\frac{\alpha_{\text{nlos}}}{2}}}{P}} {\rm d}t = \frac{2\pi \lambda_c \Gamma\l(\frac{2}{\alpha_{\text{nlos}}}\r)\l( \frac{P}{\theta}\r)^{\frac{2}{\alpha_{\text{nlos}}}}}{\alpha_{\text{nlos}}} ,
\end{align}
where $\Gamma(\cdot)$ is the Gamma function. For the general case of Nakagami fading i.e., $G_Y \sim \Gamma(g, 1/g)$ \cite{Andrews:2016aa, bai2014analysis}, $\mu(\tilde{\Pi})$ can be approximated as
\begin{align}\label{nakagami_fading}
\mu(\tilde{\Pi}) \approx \pi \lambda_c \int_{0}^{\infty} \l( 1 - \l( 1 - e^{- g(g!)^{-\frac{1}{g}}\frac{\theta}{P} t^{\frac{\alpha_{\text{nlos}}}{2}} } \r)^g \r) {\rm d}t.
\end{align}
By  \eqref{Eq:Pc:NLoS} and \eqref{Pi_bs}, $p_c(\text{NLoS})$ is the void probability with respect to the marked PPP $\tilde{\Pi}$ and is given by
\begin{equation}
p_c(\text{NLoS})  = 1 - e^{- \mu(\tilde{\Pi})}.
\end{equation}
It follows from \eqref{pc_twoterms} that the connectivity probability under the channel model with LoS and NLoS paths can be derived as
\begin{align}\label{PC_sing}
p_c = 1 - \l(1 - p_c(\text{LoS}) \r)e^{- \mu(\tilde{\Pi})},
\end{align}
where $p_c(\text{LoS})$ is given in Theorem~\ref{thm:CP1} and $\mu(\tilde{\Pi})$ is given in \eqref{density_bs}. For the special case of Rayleigh fading, substituting \eqref{density_rayfading} into \eqref{PC_sing} gives
\begin{align}\label{PC_sing_rayfading}
p_c =   1 - \l(1 - p_c(\text{LoS}) \r) \exp\l(-\frac{2\pi \lambda_c  \Gamma\l(\frac{2}{\alpha_{\text{nlos}}}\r)}{\alpha_{\text{nlos}}} \l( \frac{P}{\theta}\r)^{\frac{2}{\alpha_{\text{nlos}}}} \r).
\end{align}
For the case of Nakagami fading, we have
\begin{align}\label{PC_sing_nakagamifading}
p_c \approx    1 - \l(1 - p_c(\text{LoS}) \r) e^{- \pi \lambda_c \int_{0}^{\infty} \l( 1 - \l( 1 - e^{- g(g!)^{-\frac{1}{g}}\frac{\theta}{P} t^{\frac{\alpha_{\text{nlos}}}{2}} } \r)^g \r) {\rm d}t }.
\end{align}

\section{Connectivity of Heterogeneous Networks}\label{HetNets}
In the preceding section, the connectivity probabilities are analyzed for the single-tier mmWave network. Noting that the HetNet provides a promising approach to satisfy the rapid traffic growth by deploying short range small BSs (e.g., pico BSs) together with traditional macro BSs. The results are extended in this section to the $K$-tier mmWave HetNet. The procedure is straightforward and the details are omitted.

\subsection{Bounding Connectivity Probabilities for the $K$-tier HetNet}
It is challenging to give an exact result for $\widehat{p}_c$ due to the spatial correlation of $\{\mathcal{F}^{(k)}\}$. For analytical tractability, a lower bound on $\widehat{p}_c$ is given by selecting the largest value of the network connectivity probabilities of the $K$ tiers. Mathematically,  the  connectivity probability for the $K$-tier HetNet as defined in  \eqref{def:pc_k} is lower bounded as follows
\begin{align}
\widehat{p}_c&={\rm Pr}\l\{\bigcup_{k=1}^{K} \left( \mathcal{F}^{(k)} \cap \mathcal{B}(r^{(k)}_b) \cap \Pi^{(k)}\right)\neq \emptyset\r\}\notag\\
&\geq\max_{k \in \{1,\cdots, K\}} p_c^{(k)}, \label{eq:CPK_gene}
\end{align}
where $p_c^{(k)}\triangleq {\rm Pr} \l\{ \mathcal{F}^{(k)} \cap \mathcal{B}(r^{(k)}_b) \cap \Pi^{(k)} \neq \emptyset \r\} $ represents the connectivity probability for the $k$th tier network.

The per-tier connectivity probabilities $\{p_c^{(k)}\}$ can be bounded by modifying the lower bound on the single-tier counterpart in Theorem~\ref{thm:CP1}. Specifically, the modification involves replacing the blockage-occupancy probability $p_b$ with $(1 - \prod_{\ell=1}^k \bar{p}_b^{(\ell)})$ since the signals transmitted by a BS in the $k$th tier network can be blocked by  buildings of heights $h^{(1)},\cdots, h^{(k)}$, as illustrated in Section~\ref{system_model}. This yields the following corollary of Theorem~\ref{thm:CP1}.
\begin{corollary}[Per-tier Connectivity Probability for the $K$-tier HetNet]\label{cor:CPi_arb}
\emph{The connectivity probability for the $k$th tier network can be lower bounded as $p_{c}^{(k)}\geq \eta^{(k)}$, where $\eta^{(k)}$ is }
\begin{align}
&\eta^{(k)} = \sum_{n=0}^{\Big\lfloor\frac{r^{(k)}_{b}}{\sqrt{s}}-\frac{1}{2}\Big\rfloor^{+}}\l(1 - e^{-\pi s\lambda^{(k)}_c (n+\frac{1}{2})^2}\r) q_k^{4(n+\frac{1}{2})^2-1}\l( 1 - q_k^{8(n+1)} \r) \nn \\
&\hspace{20mm}+\l(1-e^{-\pi\lambda^{(k)}_c \l(r^{(k)}_{b}\r)^{2}}\r)q_k^{ N^{+}\l(r^{(k)}_{b}\r)-1},
\end{align}
\emph{with $q_k  = \prod_{\ell=1}^k \bar{p}_b^{(\ell)}$ and $N^{+}(r^{(k)}_{b})$ given in \eqref{num_site_sr1} and \eqref{num_site_sr2}.}
\end{corollary}
It is worth mentioning that  $q_k$ can be rewritten as $\left(1-\sum_{\ell=1}^{k}p^{(\ell)}_{b}\right)$ when  $\{p^{(k)}_{b}\}$ are small, giving a simple form. Substituting  the result in Corollary~\ref{cor:CPi_arb} into \eqref{eq:CPK_gene} leads to the first main result of this section.
\begin{theorem}[Connectivity Probability for the $K$-tier HetNet] \label{Theo:HetNets_arbitS_1region}
\emph{The connectivity probability for the $K$-tier HetNet can be lower bounded as}
\begin{align}\label{eq:CPK_arb_all}
\widehat{p}_c \geq \max_{k \in \{1,\cdots, K\}}\eta^{(k)},
\end{align}
\emph{where $\eta^{(k)}$ is given in Corollary~\ref{cor:CPi_arb}.}
\end{theorem}

Next, consider the case with a high site density ($\lambda_s \rightarrow \infty$) as before. By modifying the result in Corollary~\ref{cor:CPi_arb} in the same  way as for obtaining Theorem~\ref{thm:CP1_cont}, the connectivity probability for the $k$th tier network is given in  Corollary~\ref{them:CPi_smls}.
\begin{corollary}[Per-tier Connectivity Probability for the $K$-tier HetNet with a High Site Density]\label{them:CPi_smls}
\emph{For the case of dense sites ($\lambda_s \rightarrow \infty$),  the  connectivity probability  for the $k$th tier network is bounded as $p_c^{(k)} \geq\eta^{(k)}$ with
\begin{align}\label{eq:CPK_cont_smls}
\eta^{(k)} &  = \frac{1- e^{-\pi \lambda^{(k)}_c \left(r^{(k)}_{b}\right)^{2}\Big(1+\frac{\lambda_{s}}{\lambda^{(k)}_c}\ln\frac{1}{q_k}\Big)}}{1+\frac{\lambda_{s}}{\lambda^{(k)}_c}\ln\frac{1}{q_k}}+ \frac{1-e^{-
\pi\frac{\lambda^{(k)}_c}{\lambda_s}}}{1+\frac{\lambda^{(k)}_c}{\lambda_{s}\ln\frac{1}{q_k}}}, \nn \\
\end{align}
and  $q_k$ given in  Corollary~\ref{cor:CPi_arb}.}
\end{corollary}
Corollary~\ref{them:CPi_smls} gives a closed-form lower bound on the  connectivity probability for the $k$th tier network with dense sites. By Theorem~\ref{Theo:HetNets_arbitS_1region} and Corollary~\ref{them:CPi_smls}, the connectivity probability for the  $K$-tier HetNet with dense sites is obtained as shown below.
\begin{theorem}[Connectivity Probability for the $K$-tier HetNet with Dense Sites]\label{thm:CPK_smls}
\emph{For the case of dense sites ($\lambda_s \rightarrow \infty$), the connectivity probability for the $K$-tier HetNet can be lower bounded as}
\begin{align}
\widehat{p}_c \geq \max_{k \in \{1,\cdots, K\}} \eta^{(k)},
\end{align}
\emph{where $\eta^{(k)}$ is given  in Corollary~\ref{them:CPi_smls}.}
\end{theorem}

The lower bound can be tightened using the same approach as in the preceding section. The results are summarized in the following theorem.
\begin{theorem}\label{thm:CP1_8regions_hetnet}
\emph{An alternative lower bound on the connectivity probability for  the $K$-tier HetNet is
\begin{align}
\widehat{p}_c \geq  \max_{k \in \{1,\cdots, K\}}\l\{1 - \prod_{n=1}^{8} \l(1 - \eta^{(k,n)}\r)\r\}, \label{ieq:cover_8region_hetnets}
\end{align}
where $\{ \eta^{(k,n)} \}$ are stated as follows.
\item[--]  For $n =$ 1, 3, 5, and 7:
\begin{align}
\eta^{(k,n)}&= \l( 1 - e^{- s \lambda^{(k)}_c \l\lfloor \frac{r^{(k)}_{b}}{\sqrt{s}} -\frac{1}{2} \r\rfloor^{+}}\r)q_k^{\l\lceil \frac{r^{(k)}_{b}}{\sqrt{s}} -\frac{1}{2} \r\rceil^{+}} \nn \\
&\hspace{10mm}+ \sum_{\ell=0}^{\l\lfloor \frac{r^{(k)}_{b}}{\sqrt{s}} -\frac{1}{2} \r\rfloor^{+}}\l( 1 - e^{-s\lambda^{(k)}_c \ell} \r)\bar{q}_k q_k^\ell.
\end{align}}
\emph{\item[--] For $n =$ 2, 4, 6, and 8:
\begin{align}
\eta^{(k,n)} &=  \l( 1 - e^{- \frac{\pi s }{4}\lambda^{(k)}_c \Big(\Big\lfloor \frac{r^{(k)}_{b}}{\sqrt{s}} -\frac{1}{2} \Big\rfloor^{+} \Big)^2 }\r)q_k^{\l( \l\lceil \frac{r^{(k)}_{b}}{\sqrt{s}}-\frac{1}{2} \r\rceil^{+}\r)^2} \nn \\
&\hspace{0mm}+\sum_{\ell=0}^{\l\lfloor \frac{r^{(k)}_{b}}{\sqrt{s}} -\frac{1}{2} \r\rfloor^{+}}\l( 1 - e^{- \frac{\pi s}{4} \lambda^{(k)}_c \ell^2} \r)q_k^{\ell^2} \l(1- q_k^{2\ell+1} \r).
\end{align}
Here,  $q_k  = \prod_{\ell=1}^k \bar{p}_b^{(\ell)}$.}
\end{theorem}

An approximation of   $\widehat{p}_c$ can be obtained if  the spatial correlation of $\{\mathcal{F}^{(k)}\}$ is ignored. Mathematically,  the approximation of   $\widehat{p}_c$ is given as follows:
\begin{align}\label{eq:CPK_appro}
\widehat{p}_c &= 1 - {\rm Pr}\l\{\bigcup_{k=1}^{K} \left( \mathcal{F}^{(k)} \cap \mathcal{B}(r^{(k)}_b) \cap \Pi^{(k)}\right)= \emptyset\r\}\notag\\
&\approx 1-\prod_{k=1}^{K} {\rm Pr}\l\{ \mathcal{F}^{(k)} \cap \mathcal{B}(r^{(k)}_b) \cap \Pi^{(k)} = \emptyset  \r\} \nn \\
&= 1-\prod_{k=1}^{K}\Big(1-p_c^{(k)}\Big)\;.
\end{align}

\begin{remark}\label{remark:LinerlyGrowHetNets}
\emph{If the connectivity probability for each tier of the $K$-tier HetNet is small, namely $p_c^{(k)} \rightarrow 0$ for all $k=1,\cdots, K$,  it follows from \eqref{eq:CPK_appro} that  the network connectivity probability for the $K$-tier HetNet can be further approximated as $\widehat{p}_c \approx \sum_{k=1}^{K}p_c^{(k)}$. This quantifies the gain of multiple tiers of a HetNet that adding one more tier of BSs can linearly increase the connectivity probability. Moreover, as the number of BS tiers  increases or the site density reduces, a user lies in the service ranges of a growing number  of BSs, and user can choose one of these BSs to connect based on the metric of maximum receive power. More complex connection mechanisms can be applied such as applying bias factors on receive power for different tiers for the purpose of load balancing over the tiers (see e.g., \cite{Jeff2014LoadBal}). The connection mechanisms do not affect the current analysis. }
\end{remark}

Combining the result in Corollary~\ref{cor:CPi_arb} and \eqref{eq:CPK_appro} gives the following result.
\begin{theorem}[Connectivity Probability for the $K$-tier HetNet] \label{Theo:HetNets_appro_8region}
\emph{The connectivity probability for the $K$-tier HetNet can be approximately lower bounded by}
\begin{align}\label{eq:CPK_arb_all}
\widehat{p}_c > 1 - \prod_{k=1}^{K} \l( 1 -  \eta^{(k)} \r),
\end{align}
\emph{where $\eta^{(k)}$ is defined in Corollary~\ref{cor:CPi_arb}.}
\end{theorem}

The lower bounder in Theorem~\ref{Theo:HetNets_appro_8region} can be further  tightened using the same approach as in the preceding section, giving
\begin{align}
\widehat{p}_{c} > 1 - \prod_{k=1}^{K}\l( \prod_{n=1}^{8} \l(1 - \eta^{(k,n)} \r)  \r),
\end{align}
where $\{\eta^{(k,n)}\}$ are given in Theorem~\ref{thm:CP1_8regions_hetnet}.

\subsection{Extension to Channel Model with Both LoS and NLoS Paths for the $K$-tier HetNet}\label{sec:fadingHetnet}
The extension of the analysis in Section \ref{sec:fadingSing} to the case of the $K$-tier HetNet is straightforward.  A procedure similar to that for deriving the connectivity probability in \eqref{pc_twoterms} can be applied to obtain the connectivity probability for the $k$th tier network, denoted by $p_c^{(k)}$ and expressed as
\begin{align}\label{pc_twotermsK}
p_c^{(k)} &= 1 - \l(1 - p_c^{(k)}(\text{LoS}) \r)  \l(1 - p_c^{(k)}(\text{NLoS}) \r),
\end{align}
where $p_c^{(k)}(\text{LoS})$ ($p_c^{(k)}(\text{NLoS})$) denotes the connectivity probability with only LoS (NLoS) paths in the $k$th tier network. The expression of $p_c^{(k)}(\text{LoS})$ is given in Corollary~\ref{cor:CPi_arb} and the expression of  $p_c^{(k)}(\text{NLoS})$ is given by
\begin{align}
p_c^{(k)}(\text{NLoS})
&= {\rm Pr}\l\{\max_{Y \in \Pi^{(k)}} P^{(k)} G_Y |Y|^{-\alpha^{(k)}_{\text{nlos}}} \geq \theta^{(k)} \r\},\label{Eq:PcK:NLoS}
\end{align}
where $P^{(k)}$, $\alpha^{(k)}_{\text{nlos}}$, and $\theta^{(k)}$ are the transmit power of BS, NLoS path-loss exponent, and SNR threshold for the $k$th tier network.  Finally, the connectivity probability for the $K$-tier HetNet is lower bounded by $p_c  \geq \max_{k \in \{1,\cdots, K\}} p_c^{(k)}$. The details are ignored due to the page limitation.

\section{Simulation Results}\label{simulation}
In this section, both the Monte Carlo simulation and analytical results for characterizing the network connectivity based on the proposed random lattice building modeling are presented and compared to illustrate how the blockage effect impacts the network performance in an urban scenario. The single-tier network is considered first and followed by the $K$-tier HetNet. The results are obtained based on the following parameter setting unless stated otherwise. In the single-tier network, assume the coverage range of a BS is $r_b=150$ meters (m). The BS density is $\lambda_c = 6\times 10^{-6}$ per square meter ($1/\mathrm{m}^{2}$) and the site area is $s = 300~ \mathrm{m}^{2}$. The building occupancy probability  is $p_b = 0.3$. For the $K$-tier HetNet, we assume the number of tiers $K = 3$. The corresponding BS coverage ranges of tiers 1, 2, 3 are specified as $\{r^{(1)}_b, r^{(2)}_b, r^{(3)}_b\} = \{150, 90, 50\}$ m, respectively. The density of BSs in each tier increases, given by $\lambda^{(1)}_c = 4\times 10^{-5} ~1/\mathrm{m}^{2}$, $\lambda^{(2)}_c = 5\lambda_c^{(1)}$, and $\lambda^{(3)}_c = 10\lambda^{(1)}_c$, and the building occupancy probabilities are listed as $\{p^{(0)}_{b}, p^{(1)}_{b}, p^{(2)}_{b}, p^{(3)}_{b}\} = \{0.4, 0.1, 0.2, 0.3\}$.

Both simulation and analytical results are plotted and compared based on the proposed blockage-free region partition approaches, namely MBFC region approach and \emph{multiple sub-regions} (multi-region) partition approach. In each figure, we consider the following two setups for numerically plotting the analytical curves: \emph{i}) the analytical result with a finite site density (i.e., finite site area), \emph{ii}) the analytical result with a high site density (i.e., small site area). For the markers in the  figures, the Monte Carlo simulation results are represented using the circles: red circles stand for the simulation results of the real network scenario and black circles stand for the simulation results of the lower bounds, while the analytical results of two aforementioned approaches are plotted via the short-dashed line and solid line, respectively.

\begin{figure}[t]
\centering
\includegraphics[width=8cm]{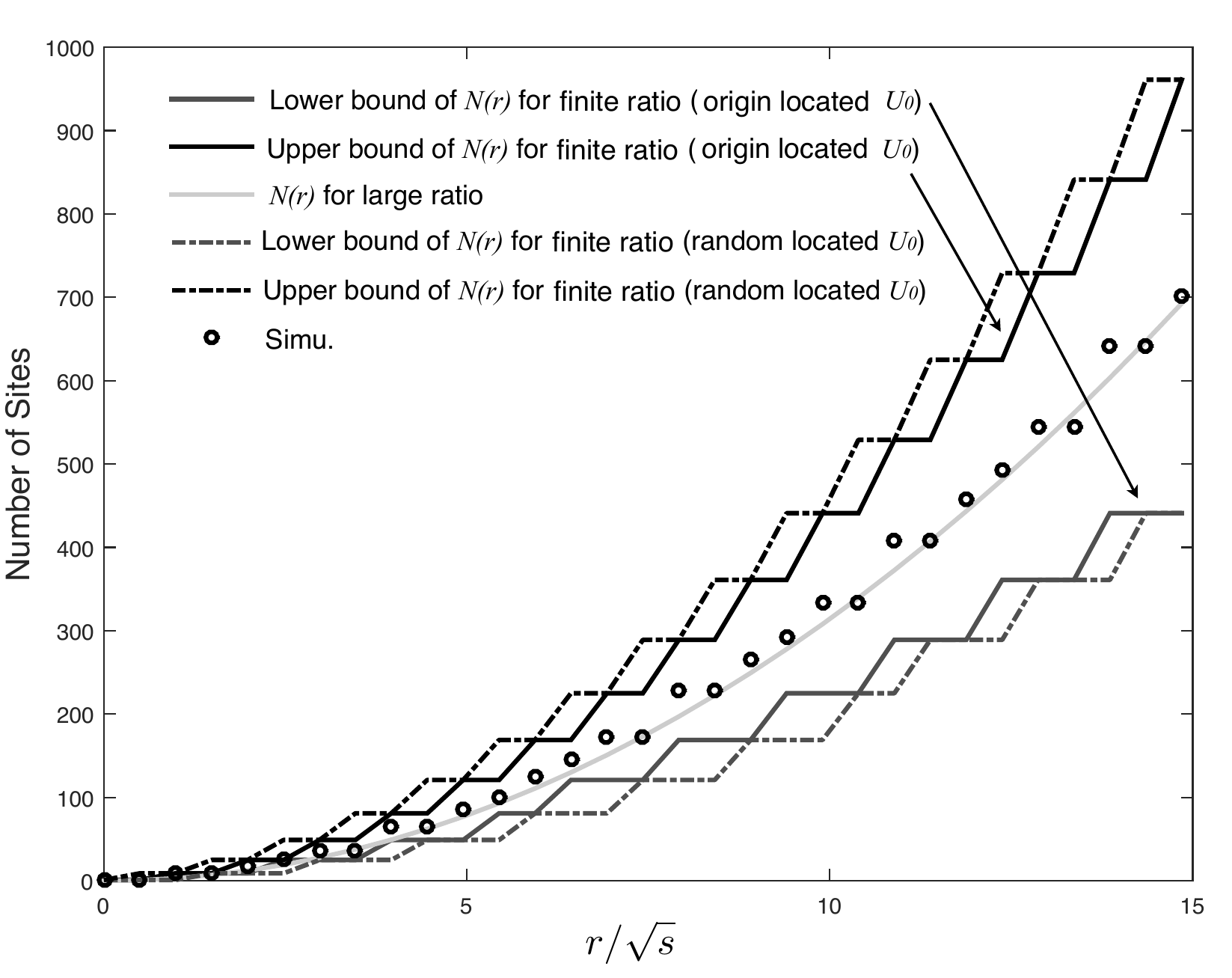}
\caption{Number of sites (fully and partially) covered by a disk with radius $r$ versus the ratio $\frac{r}{\sqrt{s}}$.}\label{sim:numsite}
\end{figure}

First, the result in Lemma~\ref{lem:num_site} is validated in Fig.~\ref{sim:numsite} displaying the curves of the number of (fully and partially) covered sites within the distance $r$ from the typical user at the origin, namely the exact $N(r)$,  its upper bound $N^{+}(r)$, lower bound $N^{-}(r)$ and approximation for large $\frac{r}{\sqrt{s}}$ given in (\ref{num_site_lr}). On one hand,  when $\frac{r}{\sqrt{s}}$ is small, the bounds and the approximation are accurate. On the other hand, when  $\frac{r}{\sqrt{s}}$ is large, the bounds are not tight but follow approximately the same scaling laws as the exact result, and the approximation  for large $\frac{r}{\sqrt{s}}$  is accurate. Combining the above two cases, the approximation of $N(r)$ appears to be tight throughout the considered range of $\frac{r}{\sqrt{s}}$. In addition, one can see that the the location of the typical user has a negligible effect on the bounds on $N(r)$ and the approximation of $N(r)$ for a large ratio $\frac{r}{\sqrt{s}}$.

\begin{figure}[t]
\centering
\includegraphics[width=8cm]{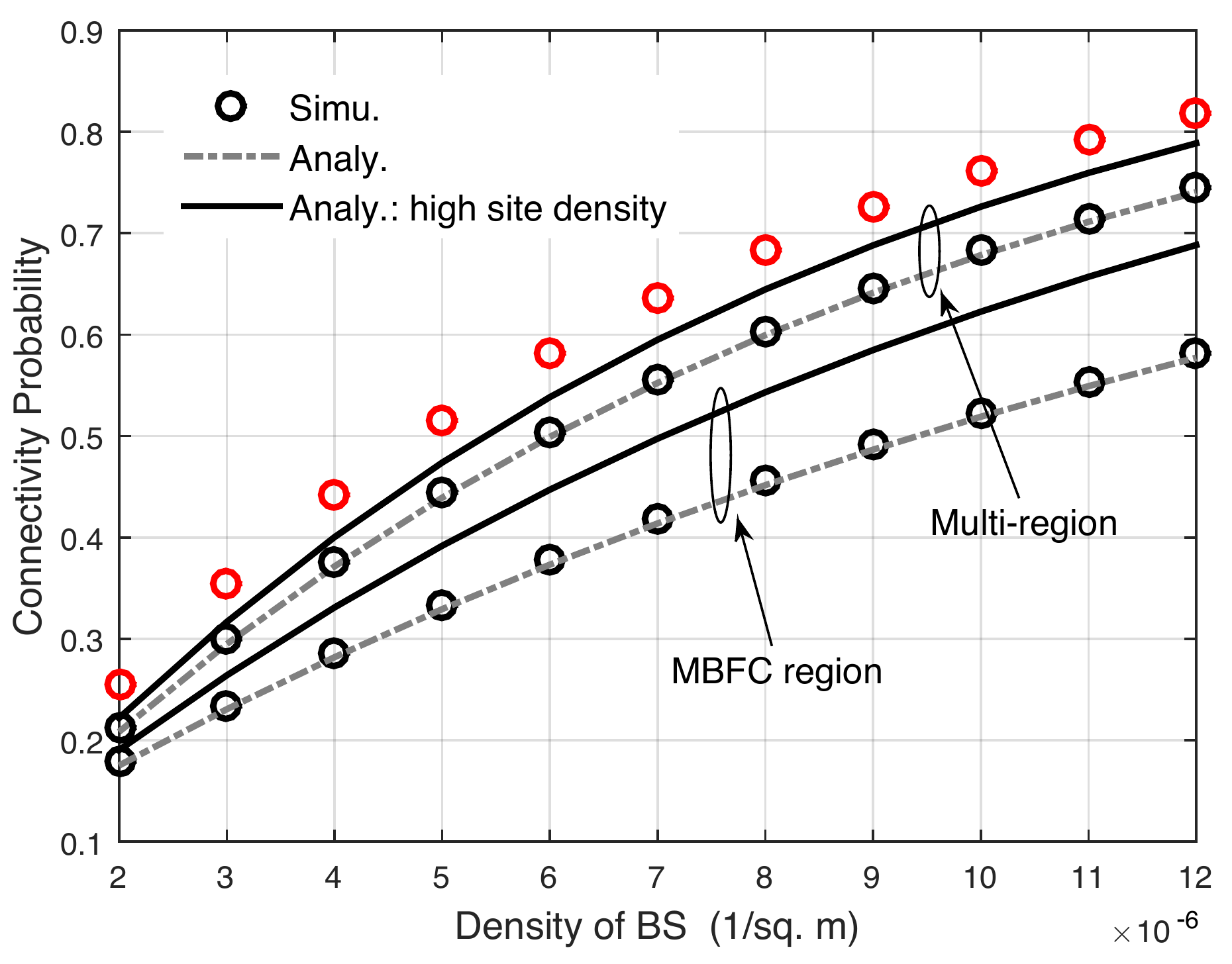}
\caption{Connectivity probability versus BS density for different blockage-free region partition approaches, MBFC region and multi-region, in the single-tier network.}\label{sim:cp_lambda_bs_1tier}
\end{figure}

In Fig.~\ref{sim:cp_lambda_bs_1tier}, the connectivity probability of the single-tier network is plotted versus the BS density for different region partition approaches. Firstly, the analytical results calculated using Theorem~\ref{thm:CP1} match the Monte Carlo simulation results (black circles) closely. It can be observed from the figure that the connectivity probability $p_c$ increases as the BS density $\lambda_c$ increases, namely, the network connectivity will benefit from a denser BS deployment since deploying more BSs leads to a larger service coverage and thus fewer coverage holes. Another observation is that the analytical result for the network with a high site density ($\lambda_s \rightarrow \infty$) given in Theorem~\ref{thm:CP1_cont} provides a tight lower bound compared with that of the network with a finite site density. The reason is that, when the site density is high, the site counting result given in (\ref{num_site_lr}) is more accurate. Moreover, the bounds have been tightened by partitioning the network into eight independent regions, which gives a small gap with respect to the simulation result of the real scenario marked by red circles. The gap seems to be an acceptable compromise between accuracy of characterizing the connectivity probability and analytical tractability.

\begin{figure}[t]
\centering
\includegraphics[width=8cm]{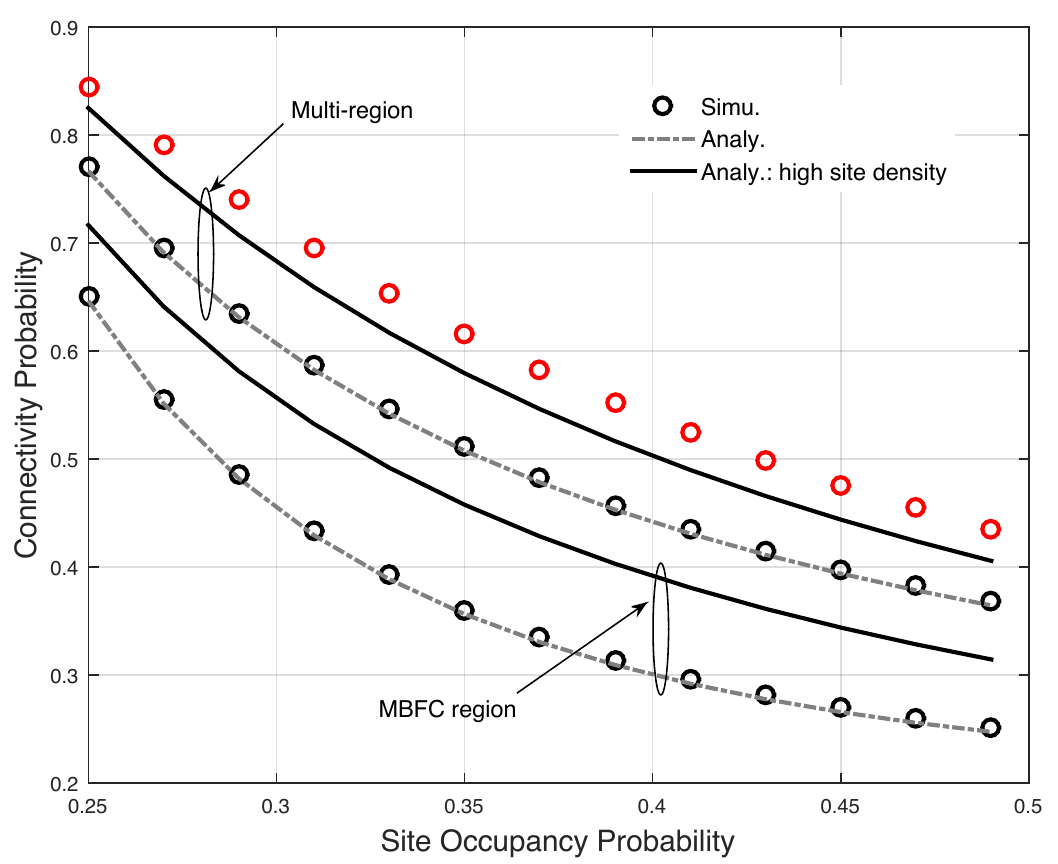}
\caption{Connectivity probability versus site occupancy probability for different blockage-free region partition approaches, MBFC region and multi-region, in the single-tier network.}\label{sim:cp_pb_1tier}
\end{figure}

\begin{figure}[t]
\centering
\includegraphics[width=8cm]{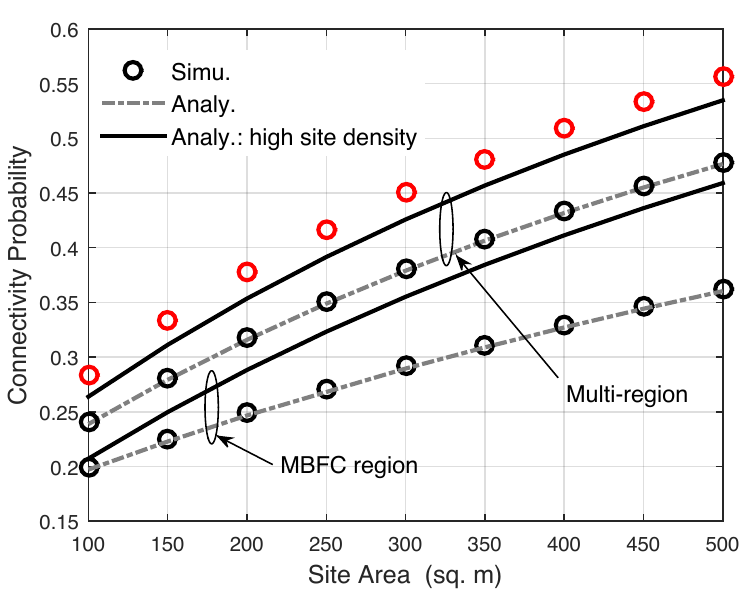}
\caption{Connectivity probability versus site area for different blockage-free region partition approaches, MBFC region and multi-region, in the single-tier network.}\label{sim:cp_s_1tier}
\end{figure}

Fig.~\ref{sim:cp_pb_1tier} shows the curves of $p_c$ in the single-tier network versus the site occupancy probability $p_b$. It is found that increasing the site occupancy probability weakens the network connectivity, and  $p_c$ degrades significantly when $p_b$ is relatively small ($0.25 - 0.35$) and converges to some value if $p_b$ becomes quite large. This is because densifying the buildings intercepts more LoS links of mmWave signals and enlarges the communication coverage holes, reducing connectivity probability.

\begin{figure}[t!]
\centering
\subfigure[Effect of BS density]{\includegraphics[width=7.5cm]{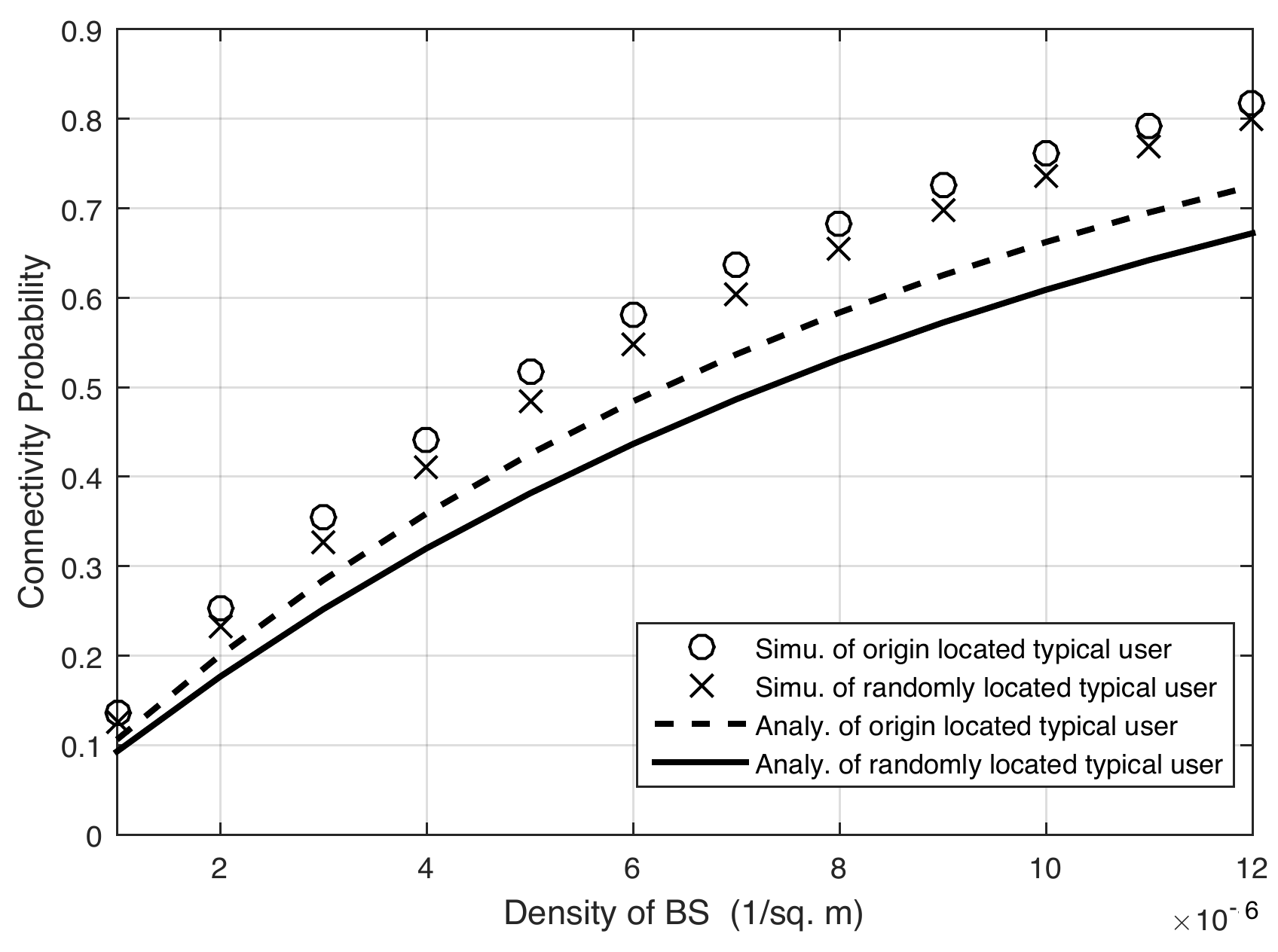}}
\subfigure[Effect of  site occupancy probability]{\includegraphics[width=7.5cm]{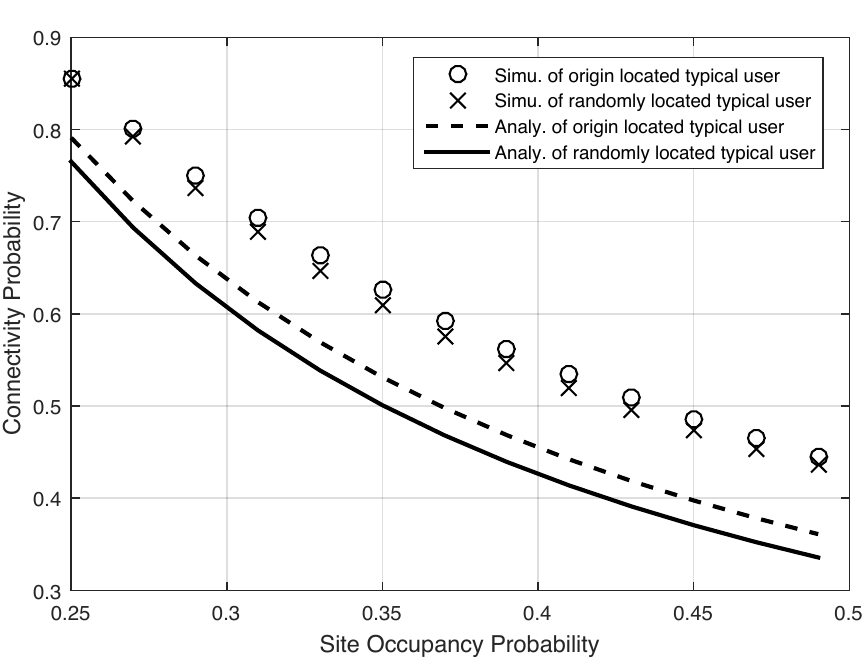}}
\caption{Connectivity probability comparison between the origin and randomly located typical user.}
\label{fig_randomUser}
\end{figure}

Next, the effect of the site density, i.e., site area $s$, on the single-tier network connectivity is investigated in Fig.~\ref{sim:cp_s_1tier}. Obviously, increasing the site area $s$, i.e., decreasing the site density $\lambda_s = \frac{1}{s}$, will dramatically improve the network connectivity. It is easy to understand due to the fact that, given the site occupancy probability $p_b$, a smaller site density yields fewer opportunities for buildings occupying the sites nearby the typical user, which is equivalent to enlarge the blockage-free region for BSs to be located thus provide better connectivity. Another observation is that the analytical curves for small $s$ are closely matched with the analytical curves as well as simulation results when the site area goes to sufficiently small, which verifies the correctness of our analytical results under the high site density assumption ($\lambda_s \rightarrow \infty$).

\begin{figure}[t]
\centering
\includegraphics[width=8cm]{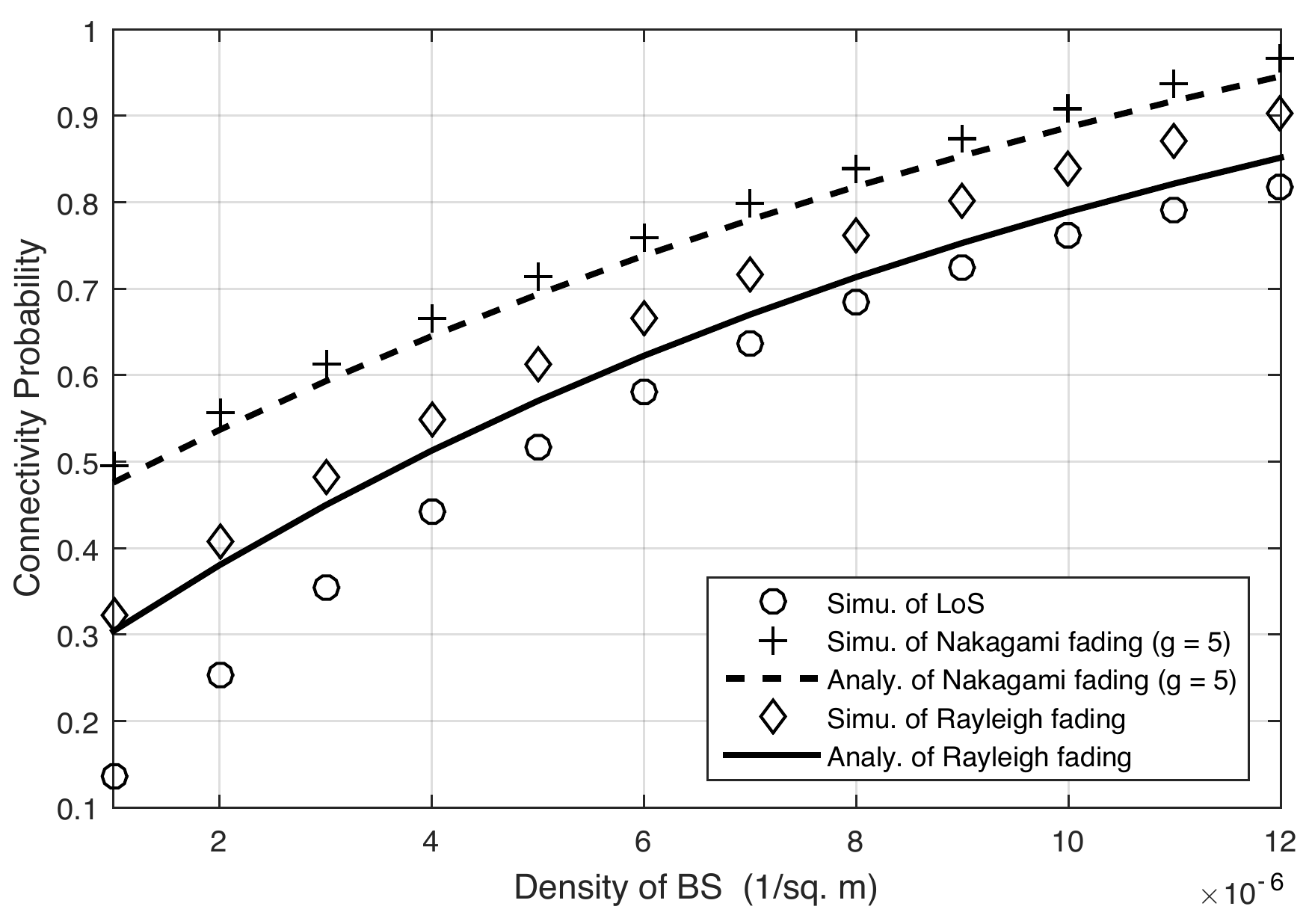}
\caption{The effect of fading in the non-LoS channel components on the coverage probability for the single-tier network.}\label{sim:fading}
\end{figure}

In Fig.~\ref{fig_randomUser}, the network performance in terms of connectivity probability is compared between the two cases of a typical user at the origin and randomly located. Both simulation and analytical results are shown. It can be observed that the random location of the typical user has a negligible effect on the connectivity probability since the corresponding simulation results closely match those for the case of the origin located user. The effects of fading  in the non-LoS channel components on connectivity probability is shown in Fig.~\ref{sim:fading}, including both simulation as well as analytical results. The small-scale fading is modelled as a Nakagami random variable with parameter $g$, i.e., $G_Y \sim \Gamma(g, 1/g)$ and other parameters are fixed as: $\theta = 5$ dB, $P = 25$ dBm, and $\alpha_{\text{nlos}} = 4$. Based on the analytical results in \eqref{PC_sing_rayfading}  and \eqref{PC_sing_nakagamifading} in Section~\ref{sec:fadingSing} and the lower bound of $p_c(\text{LoS})$ in Theorem~\ref{thm:CP1}, we can see that the simulation results can be well lower bounded and approximated by the analytical results given in the Rayleigh fading case and the Nakagami fading case, respectively. Moreover, it is observed that exploiting the non-LoS components can improve  the coverage probability. The connectivity probability reduces when $g$ decreases  (e.g., Rayleigh fading) since  the LoS  signals are weaker and more non-LoS paths exist.

\begin{figure}[h]
\centering
\includegraphics[width=8cm]{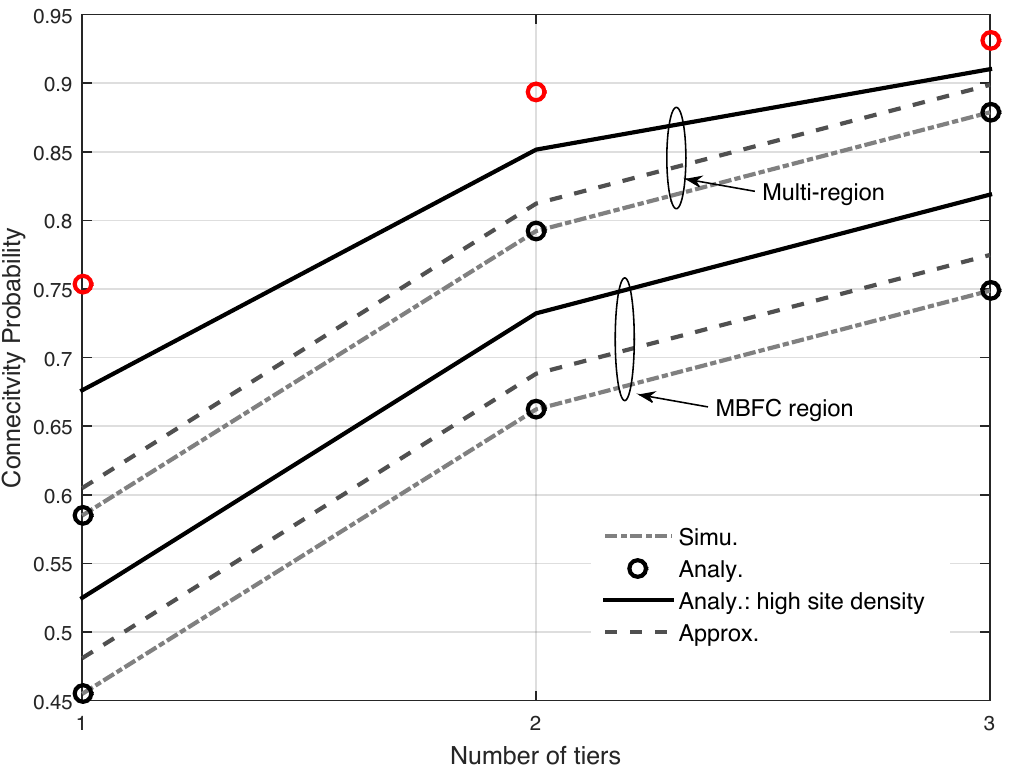}
\caption{Connectivity probability versus number of tiers for different blockage-free region partition approaches, MBFC region and multi-region, in a $K$-tier HetNet.}\label{sim:cp_k_ktier}
\end{figure}

\begin{figure}[t]
\centering
\includegraphics[width=8cm]{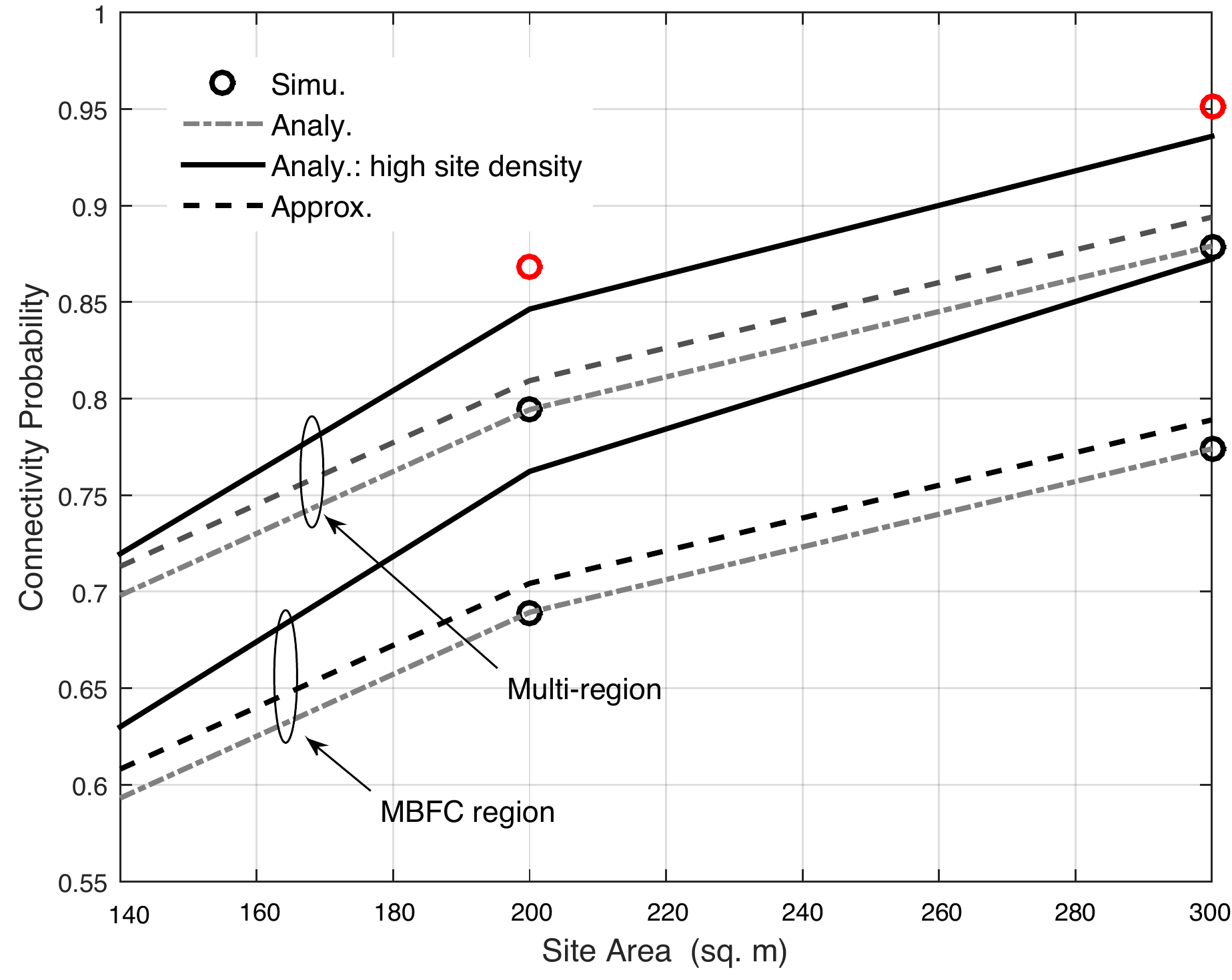}
\caption{Connectivity probability versus site area for different blockage-free region partition approaches, MBFC region and multi-region, in a $K$-tier HetNet where $K=3$.}\label{sim:cp_s_ktier}
\end{figure}

Consider the connectivity probability in the $K$-tier HetNet where $K=3$. The curves of connectivity probability $\widehat{p}_c$ are depicted against the number of tiers in Fig.~\ref{sim:cp_k_ktier}. The approximation of connectivity probability by neglecting correlations among multiple tiers is also plotted for comparison, marked as long-dashed line. As stated in the figure, $\widehat{p}_c$ is observed to grow approximately linearly with the number of tiers and saturate when $K$ is large. The bounds are remarkably tightened by partitioning the network into multi-region. Fig.~\ref{sim:cp_s_ktier} displays the curves of the network connectivity versus the site area, or site density, in the 3-tier HetNet. The connectivity probability grows when the site area $s$ increases. Moreover, the analytical results based on the high site density assumption asymptotically approach to the simulation results.


\section{Conclusion Remarks}\label{conclusion}

In this paper, we make the first attempt to investigate the blockage effect on the connectivity of mmWave networks in a Manhattan-type urban region by modeling buildings and BSs as a random lattice and a PPP, respectively. Applying the random lattice and stochastic geometry theories, different lower bounds on the connectivity probability are derived as functions of  buildings' size and the probability of a lattice cell being occupied by a building as well as BS density and transmission range. In addition, lower bounds on the asymptotic connectivity probability are also derived for cases of dense buildings. Moreover, the analysis is generalized to investigate the effect of non-LoS paths on network coverage. Last, the results are extended to HetNet. The analytical results  reveal key impacts of the building parameters and BS parameters on the connectivity of mmWave networks and provide useful guidelines for practical mmWave network deployment and performance evaluation, such as the choices for BS parameters (coverage range and density) to guarantee the connectivity of mmWave networks under given building parameters (site size, density and occupancy probability) and estimating the connectivity of mmWave networks given BS and building parameters.

This work opens up several directions for future research. In particular, taking the effects of interference and multiple antennas  into account poses interesting research opportunities. The heterogeneity in blockage objects (e.g., buildings, news stands, and billboards) can be accounted by modifying the current lattice model for buildings to be one superimposing multiple random lattices with different densities and extending the current analytical approaches.

\appendix

\subsection{Proof of Lemma \ref{lem:area_dist}}\label{proof:area_dist}
Characterizing the PMF of the random variable $R$ is equivalent to derive the area distribution of MBFC region. Denote the area measure by $\mathcal{A}(\cdot)$. Note that $\mathcal{A}(\mathcal{B}(R))$ represents the area of $\mathcal{B}(R)$, we have the following equivalent relation: ${\rm Pr}\left(R=r_n \right)={\rm Pr}\left(\mathcal{A}(\mathcal{B}(R))= \pi r_n^2 \right)$. For clear explanation, let $\mathcal{E}(r)$ be the event that there is no occupied site within $\mathcal{B}(r)$ and let $\mathcal{\bar{E}}(r)$ denote the complement event. Then, using the total probability rule gives the following result,
\begin{align}
&{\rm Pr}\left(\mathcal{A}(\mathcal{B}(R)) = \pi r_n^2 \right)\notag\\
=&{\rm Pr}\left(\mathcal{E}(r_n),\mathcal{\bar{E}}(r_{n+1})|\mathcal{E}(0)\right)\notag\\
=&{\rm Pr}\left(\mathcal{E}(r_n)|\mathcal{E}(0)\right)-{\rm Pr}\left(\mathcal{E}(r_{n+1})|\mathcal{E}(0)\right)\notag\\
=&{\rm Pr}\left(\mathcal{E}(r_n),\mathcal{E}(0)\right)-{\rm Pr}\left(\mathcal{E}(r_{n+1}),\mathcal{E}(0)\right),
\end{align}
where the last step is obtained due to the condition that ${\rm Pr}\left(\mathcal{E}(0)\right)=1$. Combining the facts that the probability that each site is not occupied is $\bar{p}_{b}$ and the occupancy of each site is independent gives the following result,
\begin{align}\label{pmf_r}
{\rm Pr}\left(R=r_n \right)=\bar{p}_{b}^{N(r_n)-1}-\bar{p}_{b}^{N(r_{n+1})-1}, n = 0, 1, \cdots. \nn
\end{align}
It is easy to verify that the number of sites within the MBFC region $\mathcal{B}(r_n)$ is $(2(n+\frac{1}{2}))^2$. Substituting $N(r_n) = ( 2(n+\frac{1}{2}) )^2$ into the above equation yields the final result.

\bibliographystyle{ieeetr}

\end{document}